\documentclass{siamart1116}
\usepackage[utf8]{inputenc}
\usepackage[T1]{fontenc}
\usepackage{graphicx}
\usepackage{cite}
\usepackage{cleveref}
\usepackage{amssymb}
\usepackage{mathabx}
\usepackage{esint}
\usepackage{enumitem}
\usepackage{units}
\usepackage{hyperref}
\usepackage{mathtools}

\usepackage{multirow,verbatim}
\usepackage[labelformat=simple]{subcaption}
\usepackage{color}

\newcommand{\cb}{{\scriptscriptstyle \mathrm{CB}}}
\newcommand{\transpose}{{\mathsf{T}}}

\newcommand{\cL}{\mathcal{L}}
\newcommand{\cR}{\mathcal{R}}
\newcommand{\cA}{\mathcal{A}}
\newcommand{\cX}{\mathcal{X}}
\newcommand{\cY}{\mathcal{Y}}
\newcommand{\cB}{\mathcal{B}}

\newcommand{\bbZ}{\mathbb{Z}}
\newcommand{\bbR}{\mathbb{R}}

\newcommand{\bx}{\mathbf{x}}

\newcommand{\bY}{\mathbf{Y}}
\newcommand{\ba}{\mathbf{a}}

\newcommand{\bg}{\mathbf{g}}
\newcommand{\bh}{\mathbf{h}}
\newcommand{\bk}{\mathbf{k}}

\newcommand{\bn}{\mathbf{n}}

\newcommand{\br}{\mathbf{r}}
\newcommand{\bu}{\mathbf{u}}
\newcommand{\bv}{\mathbf{v}}

\newcommand{\bF}{\mathbf{F}}

\newcommand{\bU}{\mathbf{U}}
\newcommand{\bgamma}{\boldsymbol{\gamma}}

\newcommand{\bxi}{\boldsymbol{\xi}}
\newcommand{\brho}{\boldsymbol{\rho}}

\newcommand{\bzero}{\boldsymbol{0}}

\newcommand{\mE}{\mathrm{E}}
\newcommand{\mI}{\mathrm{I}}
\newcommand{\mR}{\mathrm{R}}

\newcommand\1{{\mathds{1}}}

\newcommand{\be}{\mathbf{e}}

\def\XXint#1#2#3{{\setbox0=\hbox{$#1{#2#3}{\int}$ }
\vcenter{\hbox{$#2#3$ }}\kern-.6\wd0}}

\numberwithin{theorem}{section}
\newtheorem{remark}[theorem]{Remark}

\title{Energy minimization of 2D incommensurate heterostructures}

\author{Paul Cazeaux%
\thanks{Department of Mathematics, University of Kansas, Lawrence, KS 66049 (\email{pcazeaux@ku.edu})}
\and
Mitchell Luskin%
\thanks{School of Mathematics, University of Minnesota, Minneapolis, MN 55455 (\email{luskin@umn.edu}, \email{massatt067@umn.edu})}
\and
Daniel Massatt%
\footnotemark[2]
}
\date{\today}

\numberwithin{equation}{section}

\begin{document}
\maketitle
\begin{abstract}
We derive and analyze a novel approach for modeling and computing the mechanical relaxation of incommensurate 2D heterostructures.  Our approach parametrizes the relaxation pattern by the compact local configuration space rather than real space, thus bypassing the need for the standard supercell approximation and giving a true aperiodic atomistic configuration.  Our model
extends the computationally accessible regime of weakly coupled bilayers with similar orientations or lattice spacing, for example materials with a small relative twist where the widely studied large-scale moir\'e patterns arise~\cite{Kim3364,KimRelax18}.  Our model also makes possible the simulation of multi-layers for which no inter-layer empirical atomistic potential exists, such as those composed of MoS$_2$ layers, and more generally makes possible the simulation of the relaxation of multi-layer heterostructures for which a planar moir\'e pattern does not exist.

\end{abstract}




\section*{Introduction}
Two-dimensional (2D) crystals have been intensely investigated both experimentally and theoretically since graphene was exfoliated from graphite \cite{Novoselov666}.
Graphene has been shown to have exceptional mechanical strength and electrical conductivity, but its lack of a band gap limits its technological application \cite{neto2009electronic}.
Many other two-dimensional crystals such as the insulator h-BN and semiconductor MoS$_2$ have since been experimentally realized.
More recently, physicists have developed the ability to stack one layer on another with a twist angle controlled to the scale of .1 degree with the goal of creating 2D materials with desired electronic, optical, and mechanical properties \cite{Geim2013}.

Although each layer is crystalline and has a periodic structure when isolated, the ground state mechanical configuration of weakly coupled multi-layer 2D systems is generally no longer periodic when the individual layers have incommensurate lattice constants such as when graphene is stacked on MoS$_2$ or when one 2D crystal such as graphene or MoS$_2$ is stacked and rotated on a like layer~\cite{MassattDOS16,Kubo2017}.
The periodic configuration of each layer is no longer a mechanical ground state when so stacked and their relaxation may create long range modulations along the moiré pattern that strongly affect electronic properties~\cite{nam2017lattice,KimRelax18}, such as the recently observed superconductivity at so-called magic angles~\cite{Cao2018superconductivity,wen2018dihedral,carr2019minimal}.

To study mechanical relaxation for a particular twist angle or set of lattice constants, previous works have approximated the incommensurate system using an artificial strain or twist to force the system to be periodic on a supercell.
In~\cite{Wijk2015}, for example, the relaxation of bilayer graphene on supercells was studied using interatomic potentials, REBO for intralayer coupling and Kolmogorov-Crespi for interlayer coupling.
This approach is no longer feasible for small twist angles because the length scale of the supercell approximation scales inversely to the twist angle, in addition to restricting the set of configurations which can be probed~\cite{2DPerturb15}.

Other previous works have proposed or derived continuum models on moir\'e cells or supercells by imposing \textit{a priori} periodic boundary conditions, bypassing the need to strain the system into periodicity as for atomistic models.
It is, however, not clear how to rigorously link such periodic continuum deformations with the underlying aperiodic atomistic configurations for incommensurate heterostructures.
In~\cite{Dai2016}, a continuum elastic model was used for intralayer interactions while a generalized stacking fault energy (GSFE) for bilayer graphene was used to capture interlayer interactions, and in~\cite{zhang2017energy} intralayer interactions in bilayer graphene were also modeled by continuum elasticity and the intralayer interactions were modeled by a disregistry energy calculated from the Kolmogorov-Crespi interlayer atomistic potential.

In~\cite{Espanol2d}, a discrete-to-continuum model of Ginzburg-Landau type for the approximation of the in-plane and transverse displacement of a deformable square lattice weakly coupled to a rigid square lattice was derived for an empirical potential including harmonic springs, torsion, and dihedral terms for intralayer interactions and a classical Lennard-Jones term to model the weak interlayer interactions.
An asymptotic derivation was given of a continuum intralayer and interlayer energy, and numerical simulations were given of planar moir\'e relaxation patterns.

This paper presents a new conceptual approach to modeling and computing the mechanical relaxation pattern of general weakly coupled incommensurate deformable multi-layers by parametrizing the relaxation pattern by local configuration space rather than real space.
    This concept of \textit{hull}, a compact parametrization of all possible local environments, was originally introduced for the study of quantum electronic models of general aperiodic solids~\cite{bellissard1994noncommutative,bellissard2003coherent,prodan2012quantum} and recently extended to incommensurately stacked multi-layer assemblies~\cite{Kubo2017, MassattDOS16}.
The local configuration (disregistry) is given by the projection of the atomic positions of each layer onto the unit cell of the other layers.
This method is thus not limited to periodic multi-layer configurations or by a supercell approximation and provides the first representation of {\em aperiodic} atomistic configurations.
The 2D theory presented in this paper was explored in a 1D model in~\cite{Cazeaux2016} where a rigorous proof using Aubry-Mather theory is given that energy minimizers can be parameterized by local configuration.
An application of this approach to small angle twisted bilayer graphene and MoS$_2$ is given in~\cite{carrrelax}.

    Our main result is the derivation of an elastostatics model~\eqref{def:CBEnergy}-\eqref{def:CBModel} for the relaxation of vertical stacks of any number of incommensurate, weakly coupled deformable layers.
    For the particular case of bilayer heterostructures, this model reduces to the well-posed variational problem~\eqref{eq:EulerLagrange} for the continuum displacement field on a periodic moir\'e domain even for aperiodic atomistic configurations, quite similar to existing real-space continuum models~\cite{Dai2016,zhang2017energy,Espanol2d} albeit obtained by a very different route.

    In section~\ref{sec:geometry}, we describe the geometry of multi-layered structures, and in section~\ref{sec:elastic}, we propose an atomistic model for the relaxation of weakly coupled incommensurate layers and study the aliasing of the disregistry causing the moiré effect.
    In section~\ref{sec:kin}, we present the continuum interpolant in configuration space leading to a novel multilayer elastostatics model using the Cauchy-Born approximation, and in section~\ref{sec:analysis} we give an analysis of this model for bilayers.
    In section~\ref{sec:numerics}, we present our numerical approximation and give computational results for twisted bilayers.


\section{Geometry of multi-layered structures}\label{sec:geometry}
Two dimensional heterostructures are vertical stacks of a few crystalline monolayers, which typically have different periodicities.
Due to the weak van der Waals nature of the interactions between layers, little relaxation usually take place as the layers are mechanically stacked on top of each other, rather each layer essentially keeps the crystalline structure it possesses as an isolated monolayer.
The resulting assembly is in general an aperiodic system with long-range order.
This order is usually most noticeable when the lattices have close, but slightly different spacing and/or orientation, creating large scale moir\'e patterns.
In this case, an ordered, locally collective relaxation motion of atoms can take place as the layers maximize the area of energetically favorable stacking (See Figure \ref{fig:moire_relax}).

\begin{figure}
    \centering
    \begin{subfigure}[t]{.49\textwidth}
        \centering
        \includegraphics[width=.95\textwidth]{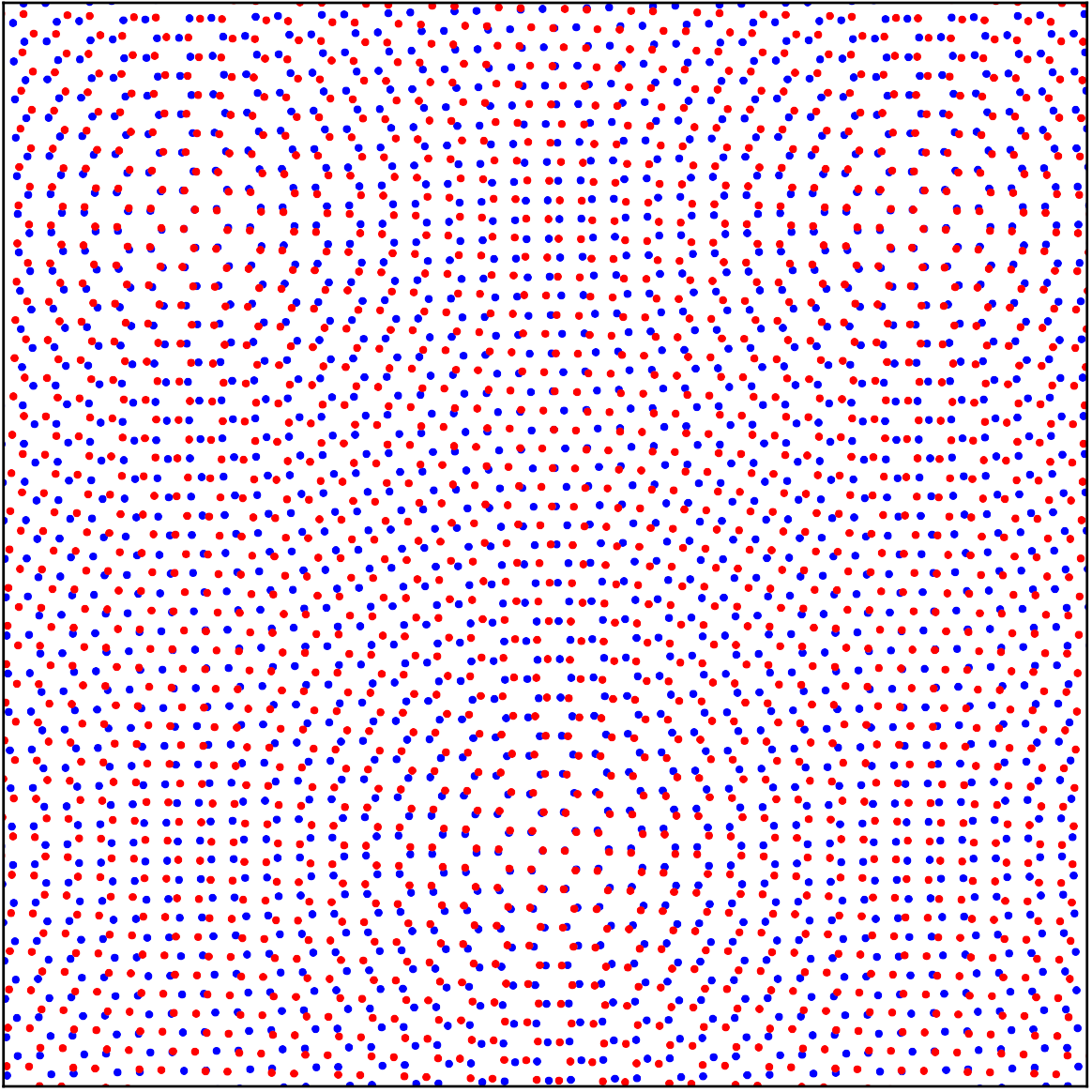}
        \caption{}
    \end{subfigure}
    \begin{subfigure}[t]{.49\textwidth}
        \centering
        \includegraphics[width=.95\textwidth]{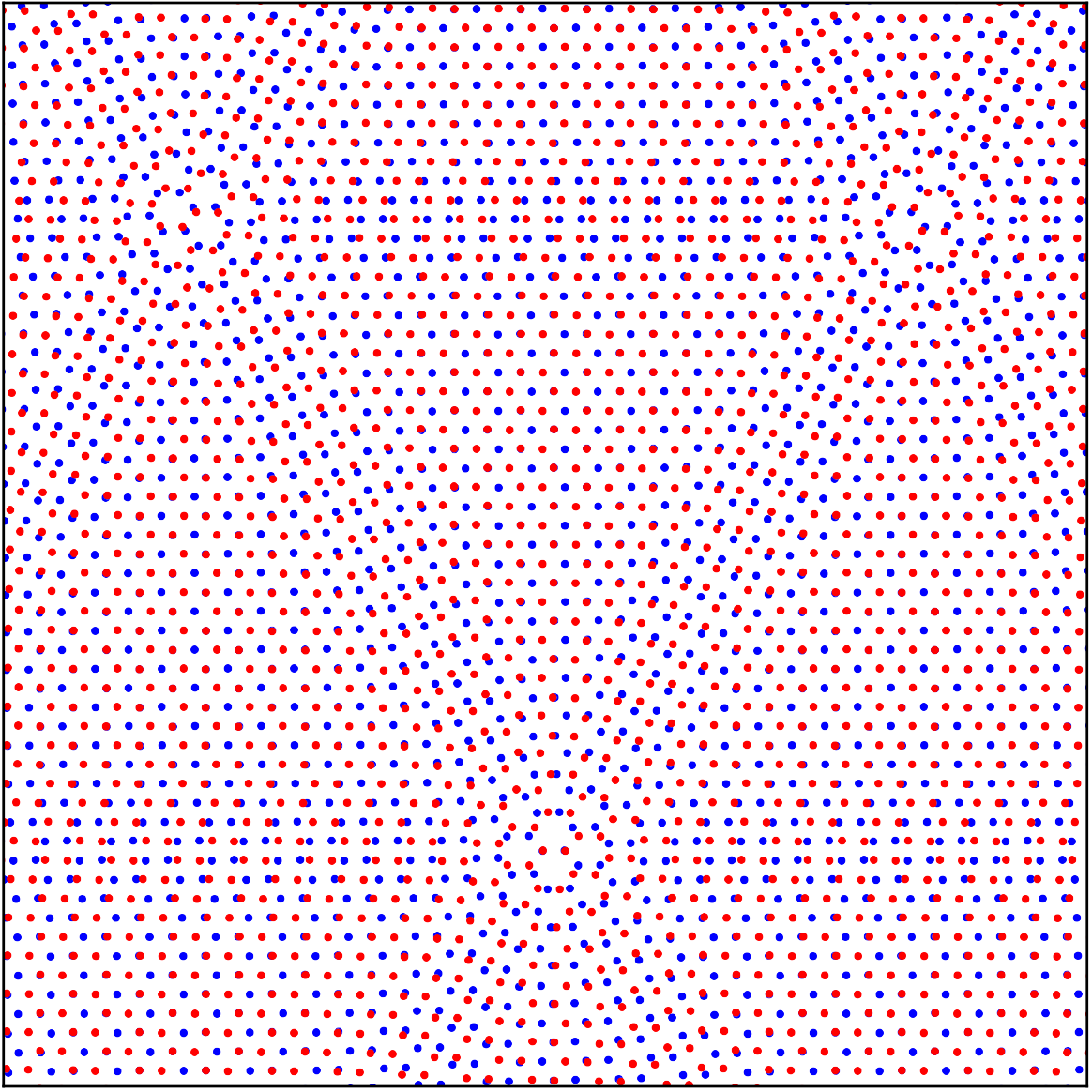}
        \caption{}
    \end{subfigure}
    \caption{Illustration of the collective, registry-driven relaxation of atomic moire patterns. (a) Unrelaxed atomic positions of G/G bilayer with a $3^\circ$ twist angle. (b) Domain structure after relaxation. The inter-layer interaction forces are enhanced by $\times 100$ for visualization purpose.}
    \label{fig:moire_relax}
\end{figure}

\subsection{Notation}

A systematic model of the relaxation of such heterostructures starts by a rigorous depiction and parameterization of the particular geometry of unrelaxed ideal multilayered structures, which will serve as a reference configuration in the elastic modeling of the relaxation phenomenon. We summarize in this section the findings of~\cite{Kubo2017}, where the geometry of these perfect multilayer structures was analyzed and discussed in detail.

We consider a 3-dimensional system of $p$ parallel 2-dimensional periodic atomic layers, denoted $\cL_j \subset \bbR^3$, $j = 1 \dots p$.
We denote by
\begin{itemize}
    \item $(\be_1, \be_2, \be_3)$ an orthonormal basis of the physical space such that each layer is perpendicular to $\be_3$; from now on, points in physical space will be identified with their cartesian coordinates $(x, y, z)^\transpose$ associated with this basis;
    \item $h_j$ the third coordinate, or height of the center of layer $j$. We assume that $0 = h_1 < h_2 < \dots < h_p$ without loss of generality;
    \item $\cR_j$, $j = 1 \dots p$ the 2-dimensional periodic lattice of layer $j$;
    \item $\mE_j$ the matrix in $\bbR^{2\times2}$ whose columns form a basis generating the layer $\cR_j$, that is,
    \begin{equation}\label{eq:lattice}
         \cR_j := \mE_j \bbZ^2 \subset \bbR^2;
    \end{equation}

    \item $\Gamma_j := \bbR^2 / \cR_j$ the quotient of $\bbR^2$ by the discrete lattice $\cR_j$, which has the topology of a 2-dimensional torus and can be canonically identified with the periodic unit cell $\widehat{\Gamma}_j := \mE_j [-1/2, 1/2)^2$ of layer $j$.
\end{itemize}

\subsection{Shifts and disregistry: the hull for incommensurate stackings}

The precise arrangement of the lattices $\cR_1$, $\cR_2,\ldots, \cR_p$ depends both on intrinsic characteristics, such as the nature of the lattice (hexagonal, square, etc.) or the value of the lattice constants, and on their position (disregistry) and orientation (twist angle). Depending on this arrangement, the overall structure can be periodic or aperiodic.
In general, periodicity arises when the intersection of the lattices $\cR_1 \cap \cdots \cap \cR_p$ forms a \textit{superlattice} of full rank.
In the generic case, the arrangement is not periodic,
meaning that a particular local arrangement of atoms is never exactly repeated as we translate horizontally the stack of layers.
\begin{figure}
    \begin{center}
        \includegraphics[width=.75\textwidth]{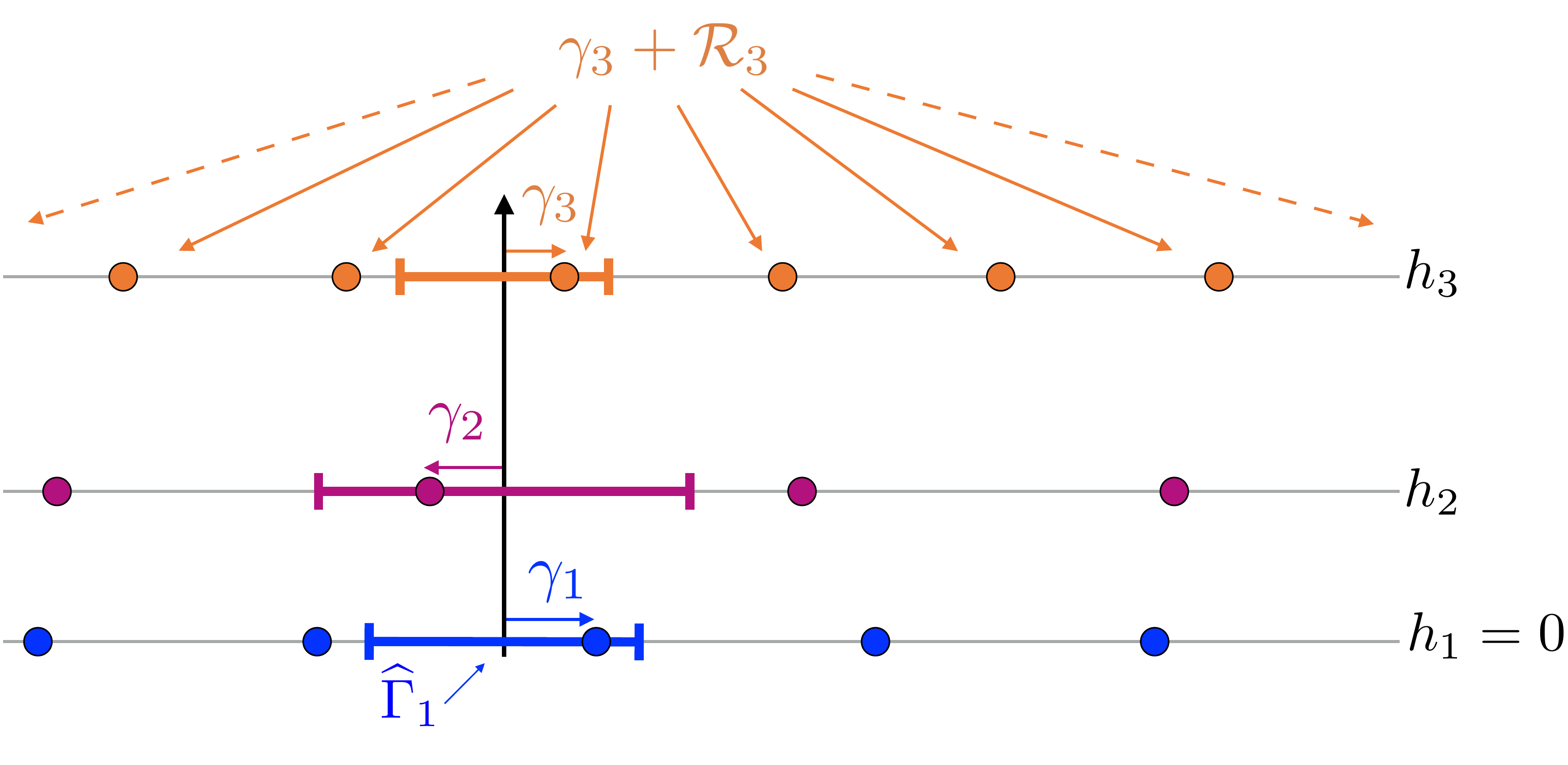}
    \end{center}
    \caption{Schematic representation of a configuration $\omega = (\bgamma_1, \bgamma_2, \bgamma_3) \in \Omega$.}
    \label{fig:hull}
\end{figure}

\paragraph{The hull}
In the context of multilayer stacks, the ensemble of possible global reference configurations
was rigorously identified in the earlier paper~\cite{Kubo2017}.
Since each layer is perfectly periodic in this reference, unrelaxed state, the position of all atoms in layer $j$ is fully identified by the knowledge of a single vector $\bgamma_j$ in its periodic unit cell $\Gamma_j$, which we call a \textit{disregistry}.
The vertices of this layer $j$ then form the shifted lattice $\bgamma_j + \cR_j$.
Once the origin is fixed, the reference configuration
can thus be uniquely parameterized as an element $\omega = (\bgamma_1, \dots, \bgamma_p)$ of the cartesian product:
\begin{equation}
    \Omega = \Gamma_1 \times \cdots \times \Gamma_p.
\end{equation}
We present in Figure~\ref{fig:hull} a one-dimensional example where an arbitrary configuration is labeled by an element of $\Omega$.
Finally, the change of coordinate origin in the plane is naturally associated with the action of the group $\bbR^2$ on $\Omega$ by translations $\mathtt{T}$ that are parallel to the layers:
\begin{equation}\label{def:Rdp1translations}
    \text{For } \mathbf{a} \in \bbR^2, \qquad
        {\mathtt T}_\mathbf{a}:  \left \{\begin{aligned}
            \Omega &\to \Omega, \\
            (\bgamma_1, \ldots, \bgamma_p) & \to \left(\bgamma_1 + \mathbf{a}, \dots, \bgamma_p + \mathbf{a} \right).
        \end{aligned}\right.
\end{equation}
This description is particularly relevant to incommensurate stacks since the spatial translates of any particular configuration form a dense subset of the full configuration space $\Omega$.
Recall indeed that the dual lattice of a full-rank lattice $\cR$ is defined by
\[
    \cR^* = \left \{ \bk \in \bbR^2 \ \vert \ \bk \cdot \bn \in \bbZ, \quad \forall \bn \in \cR \right \}.
\]
The following definition identifies the relevant incommensurability condition, first presented in~\cite{Kubo2017}:
\begin{definition}\label{definition:incommensurability}
    The collection of full-rank lattices $\cR_1, \cdots , \cR_p \subset \bbR^2$ is called incommensurate if we have for any $p$-tuple $(\bk_1, \dots, \bk_p) \in \cR_1^* \times \cdots \times \cR_p^*$,
    \begin{equation}\label{def:incommensurability}
        \sum_{j = 1}^p \bk_j = \mathbf{0} \quad \Leftrightarrow \quad \bk_j = \mathbf{0} \ \ \forall  j = 1, \dots, p.
    \end{equation}
\end{definition}
This condition is shown to be equivalent to the ergodicity property:
\begin{proposition} \label{prop:continuousHullErgodicity}
Let $\cR_1, \cdots , \cR_p$ be an incommensurate collection of lattices. Then we have the Birkhoff property: for any $f \in C(\Omega)$ and $\omega \in \Omega$,
    \begin{equation}\label{def:continuousBirkhoff}
        \lim_{r \to \infty} \frac{1}{\vert B_r \vert } \int_{B_r} f(\mathtt{T}_{-\ba} \omega) d\ba = \fint_\Omega f \qquad := \frac{1}{\vert \Omega\vert } \int_\Omega f(\omega) d\omega,
    \end{equation}
    where $B_r$ is the ball of radius $r$ centered at zero.
\end{proposition}

This ensemble of possible configurations can therefore be identified with the so-called \textit{hull}~\cite{bellissard1994noncommutative,bellissard2003coherent}, introduced previously in the study of the electronic structure of aperiodic materials.
The idea behind the hull is to picture an arbitrary position of observation in the layer plane as the origin of coordinates, and the local environment as the neighboring atoms%
\footnote{A configuration is rigorously defined as the position of all atoms in the system relative to the origin, given an arbitrary translation of the system corresponding to a change of viewpoint, typically encoded as a Radon measure in $\mathfrak{M}(\bbR^3)$.
Formally, the hull is a dynamical system $(\Omega, \bbR^2, \mathtt{T})$ where $\Omega$ is the closure of the orbit of the atomic distribution generated by the atoms of all $p$ layers under the action of $\bbR^2$ through $\mathtt{T}$. Note that the group of translations $\mathtt{T}_\ba$ with $\ba \in \bbR^2$ acts on the space of compactly supported continuous functions $\mathcal{C}_c(\bbR^3)$ naturally through $\mathtt{T}_\ba f(\bx) = f(\bx - \ba)$, and thus on the space of Radon measures $\mathfrak{M}(\bbR^3)$ through $\mathtt{T}_\ba \mu (f) = \mu(\mathtt{T}_{-\ba} f)$.}.
This is particularly relevant to model the elastic relaxation of the stack, since as we shall see in Section~\ref{sec:atomisticmodel}, the local environment of an atom typically determines its energetics due to the locality of bonding.

\paragraph{The transversal}

In the context of atomistic models of the elastic behavior of the structures, the quantity of interest is the displacement of atoms.
This leads us to introduce the so-called \textit{canonical transversal of the hull}, denoted by $X$, which is formed as the set
\begin{equation}\label{def:transversal}
    X := \bigcup_{j = 1}^p \left ( \{ j \} \times X_j \right ) \text{ where } X_j := \left \{ \omega = (\bgamma_1, \dots, \bgamma_p) \in \Omega \quad s.t.\ \bgamma_j = \bzero \right \}.
\end{equation}
Elements of $X$ label special configurations where a lattice site in a particular layer is located at the origin, with $j$ being the layer index and $\omega \in X_j$ indexing the reference configuration.
The set of allowed translations within $X$ is a discrete subset of the full translation group $\bbR^2$ called a groupoid~\cite{bellissard1994noncommutative}.
A statement of ergodicity mirroring Proposition~\ref{prop:continuousHullErgodicity} can be formulated on the transversal (see Prop. 3.5 in~\cite{Kubo2017}):
\begin{proposition} \label{prop:discreteHullErgodicity}
Let $\cR_1, \cdots , \cR_p$ be an incommensurate collection of lattices. Then we have the Birkhoff property: for any $f \in C(X)$ and $\omega \in \Omega$,
    \begin{equation}\label{def:discreteBirkhoff}
        \lim_{r \to \infty}  \frac{1}{\vert B_r \vert}  \sum_{j = 1}^p  \sum_{\ba \in \cL_{j}\cap B_r} f(j, \mathtt{T}_{-\ba} \omega)    = \sum_{j = 1}^p \fint_{X_j} \frac{f(j, \cdot)}{\vert \Gamma_j \vert}, \quad \text{where } \cL_{j} := \bgamma_j + \cR_j.
    \end{equation}
\end{proposition}

\section{Elastic models of relaxation}\label{sec:elastic}
Let us now introduce elastic models of moir\'e pattern relaxation based on the previously introduced configuration-space parameterization of the deformation. We first discuss a general atomistic energy model. Subsequently, we discuss a continuum model and how it is an effective approximation for the atomistic model.

\subsection{Atomistic kinematics}\label{sec:atomistickinematics}

Our main assumption in this paper is that there exists a one-to-one smooth correspondance between the displacement of any atom belonging to one of the layers and its initial local configuration through the elastic relaxation process.
We thus introduce a set of modulation fields with respect to the reference (unrelaxed) positions, which takes the form of a map $\bu$ : $X \to \bbR^2$. We will use the notation
\begin{equation}\label{def:atomisticmodulationfields}
    \bu_j(\omega) := \bu(j, \omega), \qquad 1 \leq j \leq p, \qquad \bu \in C(X, \bbR^2),
\end{equation}
for the evaluation of functions defined on $X$, and we understand $\bu_j(\omega)$ as the deformed position of an atom of layer $j$ initially placed at the origin in the configuration $\omega \in X_j$. The fields $\bu_j$ describe the structural modulation of each layer due to the interaction with the others.

\begin{remark}
    The key assumption behind the hull description is that there exists a one-to-one correspondance between the configuration of an atom and its modulated position in the relaxed configuration. This is a reasonable assumption due to the inherent stability of the individual lattices of the layers and the weakness of the Van der Waals interactions between the layers. In 1D, this is known to be strictly true for the Frenkel-Kontorova model for atomic chains with strictly convex nearest-neighbor atomic potentials sitting in an external incommensurate potential~\cite{AubryLeDaeron1983}.
\end{remark}
We first formulate an atomistic model with empirical multi-body interactions.
For simplicity of exposition, let us start from the simple configuration where none of the layers are initially shifted from the origin, i.e., $\omega = 0$, without loss of generality thanks to ergodicity.
By picturing temporarily the atom of layer $j$ initially placed at $\bxi \in \cR_j$ as a new origin of coordinates associated with the translated configuration $\mathtt{T}_{-\bxi} 0 \in X_j$, the displacement of this atom is prescribed to be $\bu_j(\mathtt{T}_{-\bxi} 0)$ by~\eqref{def:atomisticmodulationfields}.
We thus admit displacements and deformations of the form:
\begin{equation}\label{eq:atomdeformation}
    \bU_j(\bxi) := \bu_j(\mathtt{T}_{-\bxi} 0), \quad \bY_j(\bxi) := \bxi + \bu_j(\mathtt{T}_{-\bxi} 0), \qquad 1 \leq j \leq p, \quad \bxi \in \cR_j,
\end{equation}
where $\bu: X \to \bbR^2$ is an unknown modulation field.

\subsection{The atomistic energy}\label{sec:atomisticmodel}
We now propose a formal definition of the atomistic potential energy, which we make rigorous through the remainder of the section.
Let $\cB_j = \cR_j \setminus \{ \bzero \}$ denote intra-layer lattice directions or bonds. For discrete maps $\bv: X \to \bbR^2$ we define the finite difference stencils:
\begin{align}
    D_j\bv(j, \omega) & := \left ( \bv_j(\mathtt{T}_{-\brho} \omega ) - \bv_j(\omega)\right )_{\rho \in \cB_j}, && \text{for } (j,\omega) \in X,\\
    \Delta_{ij} \bv (j, \omega) & := \left ( \bv_i(\mathtt{T}_{-\brho} \omega ) - \bv_j(\omega)\right )_{\brho \in \bgamma_i + \cR_i}, && \text{for } (j,\omega) \in X \text{ and } i \neq j.
\end{align}
Next, we assume that there exists a many-body site energy $V_j$ for atoms in layer $j$ such that, formally, the atomistic potential energy can be presented as a sum over normalized site-energies
\begin{equation}\label{def:AtomPotE}
    \mathcal{E}^\mathrm{a}(\bu) := \sum_{j=1}^p \sum_{\xi \in \cR_j} \vert \Gamma_j \vert \Phi_{j,\mathtt{T}_{-\bxi} 0} (\bu) \quad \text{where} \ \Phi_{j,\omega} (\bu) := V_j \left [ D_j \bu , \{ \Delta_{ij} \bu \}_{i \neq j}  \right ](j, \omega).
\end{equation}
This quantity is not well-defined for general modulation $\bu$, and in fact the sum~\eqref{def:AtomPotE} does not exist in general for nonzero modulations $\bu$ which are continuous on the configuration space $X$. However the corresponding atomistic potential energy by unit area is well-defined in the thermodynamic limit:
\begin{equation}\label{def:UnitAtomPotE}
    E^\mathrm{a}(\bu) := \lim_{r \to \infty}  \sum_{j=1}^p \frac{\vert \Gamma_j \vert }{\vert B_r \vert} \sum_{\bxi \in \cR_j \cap B_r} \Phi_{j,\mathtt{T}_{-\bxi} 0} (\bu).
\end{equation}
Indeed, the Birkhoff property~\eqref{def:discreteBirkhoff} allows us to express directly the large volume spatial average as a simple average over the configuration space:
\begin{proposition}
    Given continuous many-body site energies $V_j$, $j = 1 \dots p$, the thermodynamic limit~\eqref{def:UnitAtomPotE} exists for any continuous modulation field $\bu$ as in~\eqref{def:atomisticmodulationfields} and equals:
    \begin{equation}\label{def:UnitAtomPotE2}
    E^\mathrm{a}(\bu) = \sum_{j=1}^p   \fint_{X_j} \Phi_{j, \omega} (\bu) d\omega.
\end{equation}
\end{proposition}

The definition~\eqref{def:UnitAtomPotE2} of an elastic energy for relaxed stackings is completely general and establishes the basis of an atomistic model for arbitrary incommensurate stackings which will be studied in forthcoming papers.
As a first step, we focus in this paper on showing how a simple continuum model can be formulated under certain conditions.

\begin{remark}
    In the case of multi-lattices such as graphene, a rigorous treatment should include additionally \textit{shift} fields allowing for independent motion of the sublattices~\cite{van2013symmetries}. For the sake of simplicity, we neglect this complication in this paper. This simplification can be viewed as an approximation where all atoms in the unit cell share the same displacement.
\end{remark}

\subsection{Separation of intra- and inter-layer contributions}
The physical nature of the bonds between atoms of one layer is typically different, and of a much stronger nature, than the van der Waals interactions between adjacent layers.
In particular, the bonding of atoms within an individual layer is mostly independent of the arrangement of the other layers.
Furthermore, the most significant contribution of interactions between different layers is a misfit energy due to the local disregistry between adjacent layers.

To reflect these considerations, we propose the site energy for atoms of layer $j$:
\begin{equation}\label{def:IntraInterDecomposition}
    \Phi_{j, \omega} (\bu) = \Phi_{j,\omega}^\mathrm{intra}(\bu) + \Phi_{j,\omega}^\mathrm{inter}(\bu).
\end{equation}
We make the form of each term $\Phi_{j,\omega}^\mathrm{intra}$ and $\Phi_{j,\omega}^\mathrm{inter}(\bu)$ explicit below.
The equilibrium multilayer structure is obtained by minimizing the total energy~\eqref{def:UnitAtomPotE2}
with respect to the modulation field $\bu$.
The competition between the intra-layer elastic (distortion) energy~\eqref{def:IntraSiteEnergy} and inter-layer misfit (disregistry) energy~\eqref{def:InterSiteEnergy} thus drives the deviation from the perfect (unrelaxed) moiré.

\subsection{Intra-layer elastic energy}
The intra-layer contribution $\Phi^\mathrm{intra}$ writes:
\begin{equation}\label{def:IntraSiteEnergy}
    \Phi_{j,\omega}^\mathrm{intra}(\bu) := V_j^\mathrm{intra} \left ( D_j \bu_j (\omega) \right ).
\end{equation}
Individual layers are simply two dimensional crystals, so the study of the corresponding internal elastic energy deriving from the potential $V_j^\mathrm{intra}$~\eqref{def:IntraSiteEnergy} is classical, see e.g.~\cite{ortner2013justification}.
We define the intra-layer atomistic potential energy:
\begin{equation}\label{def:IntraAtomPotE}
    E^\mathrm{a}_j (\bu_j) := \fint_{X_j } V_j^\mathrm{intra} \left ( D_j \bu_j(\omega)\right ) d \omega.
\end{equation}

\paragraph{Cauchy-Born approximate energy density}
The Cauchy-Born elastic energy density function $W$: $\bbR^{2\times2} \to \bbR \cup {\pm \infty}$ is as usual defined by
\begin{equation}\label{def:CB}
    W_j(\mathrm{F}) := V^\mathrm{intra}_j(F \cdot \cB_j ).
\end{equation}
When elastic deformations are smooth, the Cauchy-Born model is a popular approximation to the atomistic model $E^\mathrm{a}_j$. Indeed, for a smooth real-space displacement $\bU_j: \bbR^2 \to \bbR^2$ of atoms in layer $j$, $V^\mathrm{intra}_j(D_j \bU_j (\bxi))) \approx V^\mathrm{intra}_j \left (\nabla \bU_j(\bxi) \cdot \cB_j \right ) = W_j(\nabla \bU_j(\bxi) )$.
Note that the spatial gradient of the atomistic deformation only makes sense after an appropriate interpolant is constructed~\cite{ortner2013justification}.

\subsection{Inter-layer landscape: the local disregistry}

The energy landscape for inter-layer interactions is tightly linked to the notion of disregistry, i.e., the relative shift between the layers, which we now explore now more in depth.

In the case of lattices that are almost aligned the disregistry between two layers can be observed to vary slowly, on a lengthscale much larger than the lattice constants (the moiré effect) as seen on Figure~\ref{fig:moire_relax}.
A careful description of this slow variation of the disregistry in real space is key to formulating an effective model of slightly misaligned multilayer systems, and we propose the following construction.

By definition, the configuration defined above around a lattice site of layer $1$ (for example)placed at the origin is the collection of its disregistries relative to the other layers, $(\bgamma_j)_{2 \leq j \leq p}$, while $\bgamma_1 = \bzero$.
However, these vary quickly under continuous translations, see~\eqref{def:Rdp1translations}, and this simple description does not make obvious the moiré effect. Indeed, the apparent slow variation described above when moving through the layer is caused by spatial aliasing induced by discrete sampling (lattice translations).
Let us choose a point within layer $j$ as the origin in some global configuration $(\bgamma_1, \ldots, \bgamma_p) \in \Omega$, let us write $\bgamma_j = \mE_j \begin{bmatrix} \,  s & t \, \end{bmatrix}^\transpose$ with $0 \leq s,t < 1$. The lattice sites near the origin in layer $j$ are located at the four corners of a shifted unit cell of the lattice $\cR_j$:
\begin{equation}\label{def:Corners}
    \br_j^{\scriptscriptstyle 00} = \mE_j \begin{bmatrix} s\\t \end{bmatrix}, \quad
    \br_j^{\scriptscriptstyle 10} = \mE_j \begin{bmatrix} s-1\\t \end{bmatrix}, \quad \br_j^{\scriptscriptstyle 01} = \mE_j \begin{bmatrix} s\\t-1 \end{bmatrix}, \quad \br_j^{\scriptscriptstyle 11} = \mE_j \begin{bmatrix} s-1\\t-1 \end{bmatrix}.
\end{equation}
\begin{figure}[t]
    \centering
    \includegraphics[width=.6\textwidth]{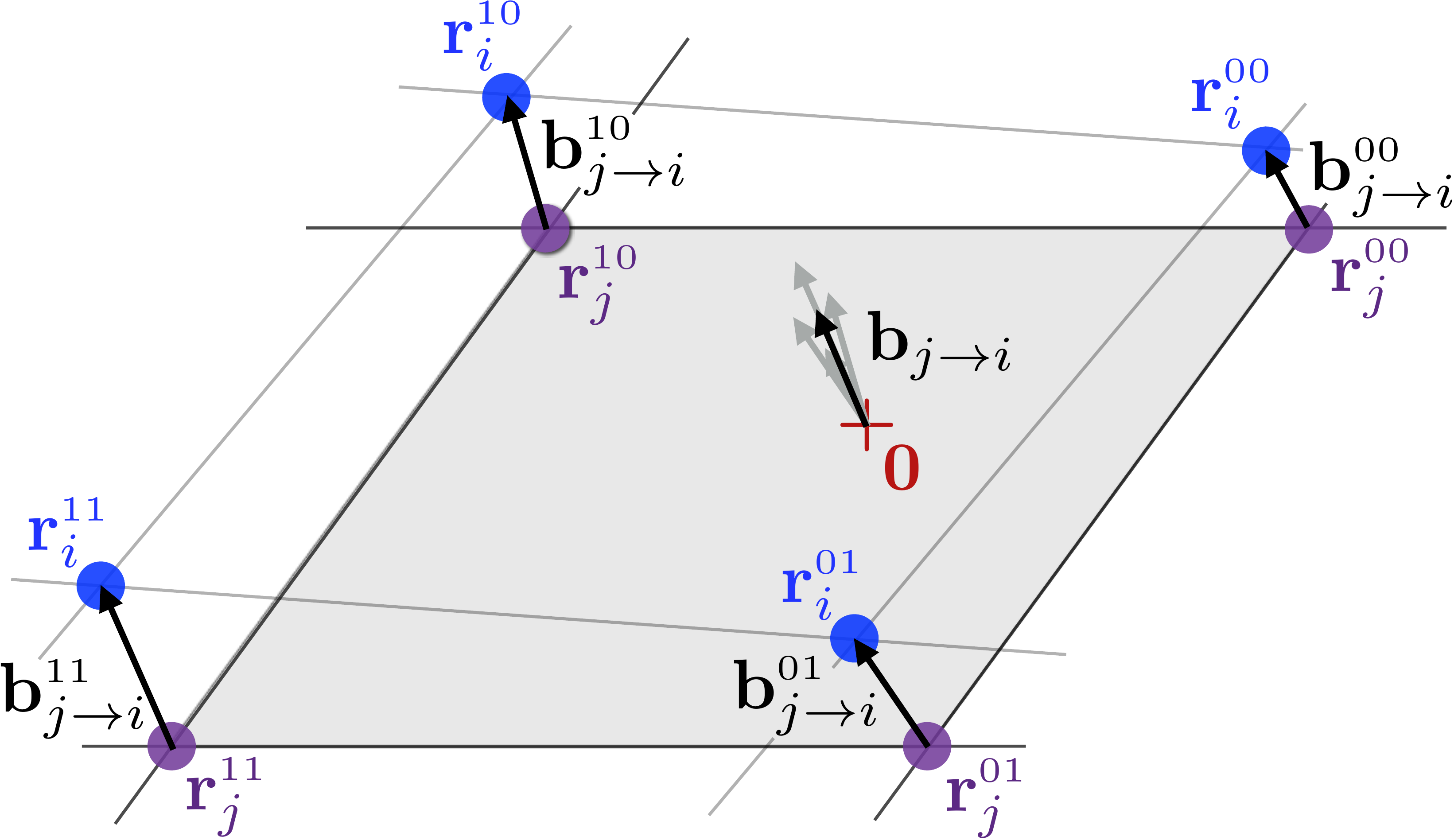}
    \caption{Illustration of the shifted unit cell used for bilinear interpolation of disregistries.}\label{fig:Corners}
\end{figure}
Similarly let us write $\bgamma_i = \mE_i \begin{bmatrix} \,  s' & t' \, \end{bmatrix}^\transpose$ and identify within layer $i$ the four lattice sites $\br_i^{\scriptscriptstyle 00}, \br_i^{\scriptscriptstyle 10}, \br_i^{\scriptscriptstyle 01}, \br_i^{\scriptscriptstyle 11}$. We measure the four disregistries between respective site pairs:
\[
    \begin{aligned}
        &\mathbf{b}_{\scriptscriptstyle j \to i}^{\scriptscriptstyle 00} := \br_i^{\scriptscriptstyle 00} - \br_j^{\scriptscriptstyle 00} = \mE_i \begin{bmatrix} s'\\t' \end{bmatrix} - \mE_j \begin{bmatrix} s\\t \end{bmatrix},
        && \mathbf{b}_{\scriptscriptstyle j \to i}^{\scriptscriptstyle 10} := \br_i^{\scriptscriptstyle 10} - \br_j^{\scriptscriptstyle 10} = \mathbf{b}_{\scriptscriptstyle j \to i}^{\scriptscriptstyle 00} + (\mE_j - \mE_i) \begin{bmatrix} 1\\0 \end{bmatrix}, \\
        &\mathbf{b}_{\scriptscriptstyle j \to i}^{\scriptscriptstyle 01} := \br_i^{\scriptscriptstyle 01} - \br_j^{\scriptscriptstyle 01} = \mathbf{b}_{\scriptscriptstyle j \to i}^{\scriptscriptstyle 00} + (\mE_j - \mE_i) \begin{bmatrix} 0\\1 \end{bmatrix},
        && \mathbf{b}_{\scriptscriptstyle j \to i}^{\scriptscriptstyle 11} := \br_i^{\scriptscriptstyle 11} - \br_j^{\scriptscriptstyle 11} = \mathbf{b}_{\scriptscriptstyle j \to i}^{\scriptscriptstyle 00} + (\mE_j - \mE_i) \begin{bmatrix} 1\\1 \end{bmatrix}.
    \end{aligned}
\]
We will assume henceforth that the lattices $\cR_1, \ldots, \cR_p$ are very close:
\begin{equation}\label{def:MoireCondition}
    \vert  \mE_j - \mE_i \vert \ll 1, \text{ for all }1 \leq i,j \leq p.
\end{equation}
Then these four disregistries differ only slightly from each other, as seen on Figure~\ref{fig:Corners}, motivating the reconstruction at the origin by bilinear interpolation in the variables $s$ and $t$ of the aliased disregistry from layer $j$ towards layer $i$:
\[
\begin{aligned}
     (1-s)(1-t) & \mathbf{b}_{j \to i}^{\scriptscriptstyle 00} + s (1-t) \mathbf{b}_{j \to i}^{\scriptscriptstyle 10}  + (1-s) t \mathbf{b}_{j \to i}^{\scriptscriptstyle 01} + s t  \mathbf{b}_{j \to i}^{\scriptscriptstyle 11} \\
    & = \mE_i \begin{bmatrix} s'\\t' \end{bmatrix} - \mE_j \begin{bmatrix} s\\t \end{bmatrix} + (\mE_j - \mE_i)  \left ( s ( 1-t) \begin{bmatrix} 1\\0 \end{bmatrix} + (1-s) t \begin{bmatrix} 0 \\ 1 \end{bmatrix} + st \begin{bmatrix} 1 \\ 1 \end{bmatrix} \right ) \\
    & = \mE_i \begin{bmatrix} s' - s \\t'  - t\end{bmatrix} \qquad \equiv \bgamma_i - \mE_i \mE_j^{-1} \bgamma_j \quad \textrm{mod}\ \mathcal{R}_i.
\end{aligned}
\]

This last quantity is well-defined in $\Gamma_i$, periodic in both variables $\bgamma_i$ and $\bgamma_j$.
We now define it as the local disregistry for any couple of layers in any configuration.
\begin{definition}
    For any configuration $(\bgamma_1, \ldots, \bgamma_p) \in \Omega$, we define the reference local disregistry $\mathbf{b}_{j \to i}$ from layer $j$ towards layer $i$ as  the quantity
    \begin{equation}\label{def:RefDisregistry}
        \mathbf{b}_{j \to i}(\bgamma_1, \ldots, \bgamma_p) = \bgamma_i - \mE_i \mE_j^{-1} \bgamma_j \in \Gamma_i.
    \end{equation}
\end{definition}
Note that when a lattice site of layer $j$ is chosen as the origin ($\bgamma_j = \bzero$), we recover the original definition of the disregistry, $\mathbf{b}_{j \to i}(\omega) = \bgamma_i$.
Furthermore, this new disregistry varies slowly under the action of translations~\eqref{def:Rdp1translations} when condition~\eqref{def:MoireCondition} is satisfied:
\[
    \mathbf{b}_{j \to i}(\mathtt{T}_\mathbf{a} \omega) = \mathbf{b}_{j \to i}(\omega) + (\mI - \mE_i \mE_j^{-1}) \mathbf{a},
\]
 confirming that~\eqref{def:RefDisregistry} captures the mesoscale pattern formation due to the moiré.

\subsection{Inter-layer misfit energy}\label{sec:misfitenergy}
The leading order contribution to the local misfit energy between two neighboring, \textit{slightly misaligned} layers $j$ and $j+1$, where $1 \leq j < p$ corresponds to the stacking-fault energy per unit area~\cite{vitek1968intrinsic} of an isolated bilayer structure formed of layers $j$ and $j+1$ and slightly rotated or stretched into periodic alignment, which we interpret as a misfit energy functional $\Phi^\mathrm{misfit}_{j,j+1}(\omega)$.

In practice, such a misfit energy has been shown in the physics literature to be most accurately described as a generalized stacking-fault energy~\cite{gong2014stacking} which can be determined from accurate Density Functional Theory calculations~\cite{marom2010stacking,constantinescu2013stacking}. Note that out-of-plane displacements can be taken into account by extending the definition of the misfit energy to include a variable inter-layer spacing (3D GSFE)~\cite{Dai2016}.

\begin{remark}
    Due to the high cost of accurate first-principles calculations including van der Waals bonding corrections, this energy density is typically computed in the approximation where both layers share the same lattice, $\mathrm{E}_a = \mathrm{E}_b$, or form a small supercell.
    Such functionals are directly available in the literature for a number of layer combinations such as G/G, G/hBN or hBN/hBN involving graphene (G) or hexagonal boron nitride (hBN) monolayers~\cite{zhou2015van}.  More recently, such a functional has been given for MoS$_2$ in \cite{carrrelax}.
    Alternatively, direct summation techniques based on empirical potentials such as the Kolmogorov-Crespi G/G potentials~\cite{kolmogorov2000smoothest} have been developed, and may allow to take into account additional local stress dependence of the stacking fault energy~\cite{zhang2017energy}.
\end{remark}

Because the local misfit energy depends only on the local reference disregistry of adjacent layers of interest, we further introduce the site-energies $\Phi^\mathrm{misfit}_{j\pm}: \Gamma_{j \pm 1} \mapsto \bbR$ localized on each layer, defined by:
\begin{equation}\label{def:GSFE}
   \Phi^\mathrm{misfit}_{j+}\left (\mathbf{b}_{j \to (j+1)}(\omega) \right ) = \Phi^\mathrm{misfit}_{(j+1)-}\left (\mathbf{b}_{(j+1) \to j}(\omega) \right ) = \Phi^\mathrm{misfit}_{j,j+1}(\omega).
\end{equation}

Note that by definition~\eqref{def:RefDisregistry} of the local reference disregistry $\mathbf{b}_{j \to (j+1)}$, these site energies satisfy the symmetry relation:
\begin{equation}\label{eq:GSFESym}
    \Phi^\mathrm{misfit}_{(j+1)-}\left (\mE_j \mathbf{s} \right ) = \Phi^\mathrm{misfit}_{j+} \left (-\mE_{j+1} \mathbf{s} \right ), \qquad \mathbf{s} \ \in\  \mE_j^{-1} \Gamma_j = \mE_{j+1}^{-1} \Gamma_{j+1} \equiv [-1/2,1/2)^2.
\end{equation}

When the layers are relaxed, we propose to write the inter-layer contribution $\Phi^\mathrm{inter}$ as:
\begin{equation}\label{def:InterSiteEnergy}
    \Phi_{j,\omega}^\mathrm{inter}(\bu) :=
            \begin{cases}
                \frac 12 \Phi^\mathrm{misfit}_{1+}\left (\mathbf{B}_{1 \to 2}[\bu](\omega)\right), & \text{if } j = 1,\\
                \frac 12 \Phi^\mathrm{misfit}_{j+}\left (\mathbf{B}_{j \to j+1}[\bu](\omega)\right) + \frac 12 \Phi^\mathrm{misfit}_{j-}\left (\mathbf{B}_{j \to j-1}[\bu](\omega) \right ), & \text{if } 1 < j < p, \\
                \frac 12 \Phi^\mathrm{misfit}_{p-}\left (\mathbf{B}_{p \to p-1}[\bu](\omega)\right), & \text{if } j = p,
            \end{cases}
\end{equation}
where $\mathbf{B}_{j \to i}[\bu]$ is a \textit{modulated} local disregistry of layer $j$ with respect to layer $i$,
taking into account the modulation of both layers.
However, the modulation of two different layers may not be defined simultaneously at the origin since $(i, \omega) \notin X$ (i.e, $\bgamma_i \neq \bzero$) in general when $(j,\omega) \in X$ (i.e, $\bgamma_j = \bzero$) for $i \neq j$. An appropriate definition will be presented below in Section~\ref{sec:ModDisregistry}.

\section{Kinematics in configuration space and continuum approximation}\label{sec:kin}

To build an effective elastic model approximating the atomistic energy~\eqref{def:UnitAtomPotE2}, two essential ingredients are now needed.
An appropriate continuum interpolant for the atomistic deformation is needed to define a spatial gradient used in the Cauchy-Born approximation~\eqref{def:CB} of the intra-layer energy.
Next, a modulated local disregistry is needed to compute the inter-layer energy~\eqref{def:InterSiteEnergy}.

\subsection{Continuum interpolant}\label{sec:continuuminterpolant}

    As noted earlier, an appropriate interpolant is needed to define a spatial gradient for the atomistic deformation.
    Continuum modulation fields
    relevant to a continuum (e.g., plate) modeling of the elastic behavior take the form of $j$ functions
    \begin{equation}\label{def:modulationfields}
        \bu^\mathrm{c}_j \in C(\Omega, \bbR^2), \qquad 1 \leq j \leq p,
    \end{equation}
    which parameterize the displacement at the origin in each of the layer planes depending on the global configuration $\omega$.

A first idea to construct an interpolant of an atomistic modulation field~\eqref{def:atomisticmodulationfields} into a continuum interpolant of the form~\eqref{def:modulationfields} is to interpolate the displacements of the closest sites to the origin as earlier for the definition of
the reference local disregistry~\eqref{def:RefDisregistry}.
Let $\omega = (\bgamma_1, \ldots, \bgamma_p) \in \Omega$ be an arbitrary configuration.
Consider a particular layer $j$ and write $\bgamma_j = \mE_j \begin{bmatrix} s \\ t \end{bmatrix}$ with $0 \leq s, t < 1$.
The nearest lattice sites in layer $j$ are located at the four points $\br_{ab}$ with $a, b \in \{ 0,1 \}$ defined as in~\eqref{def:Corners}.
Then, the bilinear interpolant of the atomic modulation of layer $j$ at the origin is
\begin{equation}\label{def:InterpolantDisplacement}
    \begin{aligned}
        \widetilde{\bu}^\mathrm{c}_j(\omega) := (1-s)(1-t)\; \bu_j(\mathtt{T}_{-\br_{00}} \omega) & + s(1-t)\; \bu_j(\mathtt{T}_{-\br_{10}} \omega) \\
                    & + (1-s)t\; \bu_j(\mathtt{T}_{-\br_{01}} \omega) + st\; \bu_j(\mathtt{T}_{-\br_{11}} \omega).
    \end{aligned}
\end{equation}
The modulation field $\widetilde{\bu}^\mathrm{c}_j$ is a continuous interpolant which allows us to lift the atomistic modulation defined on the transversal $X$ to a set of continuum modulation fields
defined on $\Omega$, for each of the layer planes.
\begin{remark}
    Alternatives could be proposed for this spatial interpolant~\eqref{def:InterpolantDisplacement}, for example using linear interpolation on any triangular mesh based on the lattice, or also higher-order schemes.
\end{remark}
A second idea is to avoid the construction of the interpolant~\eqref{def:InterpolantDisplacement} by instead
evaluating the atomistic modulation
at the (interpolated) local disregistry of layer $j$ with respect to the other layers:
\begin{equation}\label{def:InterpolantDisplacement2}
    \widehat{\bu}^\mathrm{c}_j(\omega) := \bu_j \left (\Pi_j \omega \right ),
\end{equation}
where we have introduced the configuration interpolant $\Pi_j:$ $\Omega \to \Omega$ defined by
\begin{equation}\label{def:InterpolantConfiguration}
    \Pi_j \omega := \left (\mathbf{b}_{j \to 1} (\omega), \ldots, \mathbf{b}_{j \to p} (\omega) \right ) \quad \text{so that} \quad (j, \Pi_j \omega) \in X.
\end{equation}
Any spatial interpolant such as~\eqref{def:InterpolantDisplacement} only relies on proximity in \textit{physical} space. The accuracy of a method based on its use can be certified, in the sense that it can be framed in the context of well understood atomistic-to-continuum methods~\cite{ortner2013justification}. While much easier to use, the configuration interpolant~\eqref{def:InterpolantDisplacement2} relies on proximity in \textit{configuration} space (as seen on Figure 3), so its accuracy is not so readily analyzed, depending on smoothness properties in configuration space that are not always guaranteed for solutions of the atomistic model~\eqref{sec:atomisticmodel}, see~\cite{Cazeaux2016}.

Let us highlight how the difference between the two proposed approaches is quantified by smoothness in configuration space with the following analysis.
Since the local disregistry varies slowly, at the scale of the moiré pattern,
both approaches are very close as long as the atomic modulation $\bu$ is smooth enough as a function of the configuration parameter. In particular, if $\bu_j$ is (locally) linear in $\omega$ then the two modulation interpolants~\eqref{def:InterpolantDisplacement}
and~\eqref{def:InterpolantDisplacement2} are in fact equal. This observation is the basis of the following estimate, which we prove in Appendix~\ref{sec:appendix}:
\begin{proposition}\label{prop:Interpolants}
Let $\bu \in W^{2,q}(X)$, then:
\begin{equation}\label{est:InterpolantDifference}
     \left \Vert \widehat{\bu}^\mathrm{c}_j - \widetilde{\bu}^\mathrm{c}_j \right \Vert_{L^q(\Omega)} \leq \left ( \frac{\vert \Gamma_j \vert}{2q+1} \right )^{1/q} \theta ^2 \; \left \Vert \nabla^2_\omega \bu_j \right \Vert_{L^q(X_j)}
\end{equation}
where $\nabla^2_\omega \bu_j$ is understood as a $2$-linear form for which the norm is defined as $\Vert \ell \Vert := \sup_{\vert \mathbf{h}_1 \vert = \vert \mathbf{h}_2 \vert = 1 } \left \vert \ell[\mathbf{h}_1, \mathbf{h}_2] \right \vert $ and
\begin{equation}\label{def:theta}
    \theta = \sqrt{p} \sup_{1 \leq i,j \leq p} \left \Vert \mE_i - \mE_j \right \Vert, \qquad \text{where } \Vert \cdot \Vert \text{ denotes the $\bbR^2$-operator norm}.
\end{equation}
\end{proposition}

\subsection{Evaluation of the gradient}
Let us now define a spatial gradient associated with
a given atomistic modulation field
$\bu: X \mapsto \bbR^2$. A spatial derivative does not make sense directly for atomistic quantities, but we may compute one for the configuration interpolant $\widehat{\bu}^\mathrm{c}_j$ introduced earlier~\eqref{def:InterpolantDisplacement2}. Note that this is easier than for a traditional piecewise-differentiable spatial interpolant such as~\eqref{def:InterpolantDisplacement}.

Given a global configuration $\omega \in \Omega$, in the $j$-th layer, smooth real-space fields for the displacement, $\mathbf{U}_j^{c,\omega}$, and deformation, $\mathbf{Y}_j^{c,\omega}$, maybe be evaluated at any point $\bx \in \bbR^2$ from the continuum modulation field $\widehat{\bu}^\mathrm{c}_j$ by the formulae

\[
    \mathbf{U}^{\mathrm{c},\omega}_j(\bx) := \widehat{\bu}^\mathrm{c}_j(\mathtt{T}_{-\bx}\omega), \qquad \bY^{\mathrm{c},\omega}_j(\bx) := \bx + \widehat{\bu}^\mathrm{c}_j(\mathtt{T}_{-\bx}\omega).
\]
This representation allows us in particular to rewrite the strain as a directional derivative in the hull framework.
Given the action of the translation on the hull~\eqref{def:Rdp1translations}, it is appropriate to introduce the spatial derivative on $\Omega$ as the operation
\[
    \nabla_\bx \quad := \quad - \sum_{j=1}^p \nabla_{\bgamma_j},
\]
leading us to compute:
\[
    \nabla_\bx \widehat{\bu}^\mathrm{c}_j(\omega) = - \sum_{i\neq j} [\nabla_{\bgamma_i} \bu_j]  (\Pi_j \omega)  - \sum_{i \neq j} [\nabla_{\bgamma_i} \bu_j]  (\Pi_j \omega) \cdot (-\mE_i \mE_j^{-1}).
\]

Evaluating at a point $(j, \omega) \in X$, so that $\Pi_j \omega = \omega$, we obtain
the spatial gradient of the atomistic modulation:
\begin{equation}\label{def:GradAtomDef}
    \widehat{\nabla}_\bx \bu_j (\omega) := \sum_{i \neq j} [\nabla_{\bgamma_i} \bu_j] (\omega) \cdot (\mE_i \mE_j^{-1} - \mI), \quad \text{for } (j,\omega) \in X.
\end{equation}

\subsection{Modulated local disregistry}\label{sec:ModDisregistry}

In a deformed state, the local disregistry is shifted from its reference value due to the modulation fields.
We propose the following definition for the local disregistry in the modulated atomic configuration:
\begin{definition}
    Given an atomistic modulation $\bu$,
    we define the modulated local disregistry
    $\widehat{\mathbf{B}}_{j \to i}$ of a site of layer $j$ with respect to layer $i$:
    \begin{equation}\label{def:ModDisregistry3}
        \Gamma_i \owns  \widehat{\mathbf{B}}_{j \to i}(\omega) := \mathbf{b}_{j \to i}(\omega) + \widehat{\bu}^\mathrm{c}_i \left ( \omega \right ) - \bu_j(\omega), \quad \text{for } (j,\omega) \in X.
    \end{equation}
\end{definition}

\subsection{Elastostatics model for multi-layer heterostructures}
We now build a continuum-like model \textit{on the atomistic configuration space} $X$, using the new tools introduced above:
a gradient for the atomistic modulation~\eqref{def:GradAtomDef},
and the modulated local disregistry~\eqref{def:ModDisregistry3}.
Note that the dimensionality is thus reduced compared to a true continuum description on the full hull $\Omega$ as described in Section~\ref{sec:continuuminterpolant}.

We define an elastic energy density approximating $E^\mathrm{a}_j$ as well as an overall intra-layer energy density, using the Cauchy-Born approximation~\eqref{def:CB}:
\begin{equation}\label{def:CBEnergy}
    E^\cb_j(\bu_j) := \fint_{X_j} W_j(\widehat{\nabla}_\bx \bu_j (\omega)) d\omega, \qquad \qquad E^\cb(\bu) := \sum_{j=1}^p E_j^\mathrm{c}(\bu_j),
\end{equation}
as well as an inter-layer misfit energy density:
\begin{equation}\label{def:MisfitEnergy}
        E^\mathrm{misfit}(\bu) = \frac 12 \sum_{j = 1}^{p-1} \fint_{X_j} \Phi_{j+}^\mathrm{misfit}(\widehat{\mathbf{B}}_{j\to j+1}(\omega))d\omega  + \frac 12 \sum_{j = 2}^{p} \fint_{X_j}\Phi_{j-}^\mathrm{misfit}(\widehat{\mathbf{B}}_{j\to j-1}(\omega))d\omega,
\end{equation}
where we have used the interpolated modulated local disregistry $\widehat{\mathbf{B}}_{j \to j \pm 1}$ defined earlier~\eqref{def:ModDisregistry3}.
Finally, we formulate the elastostatics model for multilayer heterostructures:
\begin{equation}\label{def:CBModel}
    \bu^* \in \mathrm{arg\ min} \left \{\  E^\cb(\bu) + E^\mathrm{misfit}(\bu) \quad \vert \quad \bu \in H^1\cap W^{1,\infty}(X; \bbR^2)\  \right \},
\end{equation}
and we understand~\eqref{def:CBModel} to be a \textit{local} minimization problem~\cite{blanc2007atomistic}
with respect to the $W^{1,\infty}$ topology on the collection of periodic tori $X$.

The compact elastostatics model constituted by equations~\eqref{def:CBEnergy}, \eqref{def:MisfitEnergy} and~\eqref{def:CBModel} for the relaxation of multi-layer heterostructures generalizes known bilayer approximations~\cite{Dai2016,zhang2017energy} to any number $p \geq 2$ of incommensurate layers without relying on approximation by periodic configurations (supercells).
The continuum relaxed displacements $\bU_j^{\mathrm{c}, \omega}$ are not periodic in any two-dimensional real space domain when $p > 2$, but the modulations $\bu_j$ are periodic in the local configuration space formed as a collection of $2(p-1)$-dimensional tori.

\begin{remark}
    Here we do not include possible external uniform/periodic forces in our formulation, although our analysis extends this case in a straightforward manner.
\end{remark}
\begin{remark}
    We also do not include out of plane bending in our presentation, though we note all the analysis could be extended to include bending.
    In this case, the intra-layer continuum model would become more complex to include curvature effects, and the modulated disregistry would need to include parallax effects due to out of plane bending as in~\cite{Dai2016}.
\end{remark}
\begin{remark}
We note that the case of a rigid substrate such as considered in~\cite{Espanol2d} can be modeled by adding the constraint $\bu_1 \equiv 0$ to the minimization~\eqref{def:CBModel}.
\end{remark}

An important question, which we will not address in the present manuscript, is whether one can show that the solutions to the elastostatics problem~\eqref{def:CBModel} are in fact close to solutions to the atomistic problem~\eqref{sec:atomisticmodel} it approximates, a classical question regarding atomistic-to-continuum approaches~\cite{ortner2013justification}. The main difficulty is to show that solutions of the discrete atomistic model are smoothly parameterized in configuration space such that Prop.~\ref{prop:Interpolants} validates the use of a configuration-space interpolant. Within the context of the one-dimensional Frenkel-Kontorova model a full analysis is possible~\cite{Cazeaux2016} using the non-perturbative Aubry-Mather theory and indicates that the modulation or hull function is smooth in the perturbative regime, but becomes fully discontinuous (Cantor-like) in the configuration space after the commensurate-incommensurate transition, indicating the limit beyond which the elastostatics model~\eqref{def:CBModel} is not expected to be an effective approximation. A similar behavior can be formally expected in 2D, but a rigorous analysis seems presently out of reach as the extension of the Aubry-Mather theory to higher-dimensional systems is still essentially an open problem (as discussed in section 2.1.1 in~\cite{su2017continuous}).


\section{Analytical study and results: bilayer case}\label{sec:analysis}
In this section, we present a first result on the existence of minimizers for the relaxation problem~\eqref{def:CBModel}.
We limit our current scope to the case of bilayers $(p = 2)$.
Note that there are additional difficulties linked with the degenerate ellipticity of the energy functional $E^\cb$ for $p \geq 3$, and we postpone to later studies their investigation.

\subsection{The moir\'e cell for bilayers}\label{sec:moirecell}
Let us fix $p = 2$. In this case, the configuration space is simply $\Omega = \Gamma_1 \times \Gamma_2$ and the transversal for each layer is simply $X_1 = \Gamma_2$, $X_2 = \Gamma_1$.
The analysis can be simplified by an appropriate change of variables: let us define the moir\'e lattice
\[
    \mathcal{R}_\mathcal{M} := (\mE_2^{-1} - \mE_1^{-1})^{-1} \bbZ^2,
\]
and the associated periodic moir\'e cell:
\begin{equation}\label{def:MoireCell}
    \Gamma_\mathcal{M} := \bbR^2 / \mathcal{R}_\mathcal{M} \equiv (\mE_2^{-1} - \mE_1^{-1})^{-1} [-1/2, 1/2)^2.
\end{equation}
We introduce next the linear mappings
\begin{equation}\label{def:MoireMappings}
    \bgamma_1: \left \{ \begin{aligned}
        \Gamma_\mathcal{M} &\to \Gamma_1, \\ \bx &\mapsto (\mE_1\mE^{-1}_2 - \mI) \bx,
    \end{aligned} \right. \qquad \qquad
    \bgamma_2: \left \{ \begin{aligned}
        \Gamma_\mathcal{M} &\to \Gamma_2, \\ \bx &\mapsto (\mE_2\mE^{-1}_1 - \mI) \bx.
    \end{aligned} \right.
\end{equation}
One checks easily that $\bgamma_1$ and $\bgamma_2$ are isomorphisms that satisfy $\Pi_1( \bgamma_1,\,0) = (0,\,\bgamma_2)$ and $\Pi_2 (0,\,\bgamma_2) = (\bgamma_1,\,0)$.
Using these as change of variables,
we introduce the new modulation unknowns
defined on the new domain $\Gamma_\mathcal{M}$:
\begin{equation}
    \bu_{\mathcal{M},1}(\bx) := \bu_1(\bgamma_2(\bx)), \qquad \qquad \bu_{\mathcal{M},2}(\bx) := \bu_2(\bgamma_1(\bx)).
\end{equation}

One of the advantages of working on the domain $\Gamma_\mathcal{M}$ is that the spatial gradient defined earlier~\eqref{def:GradAtomDef} corresponds by the chain rule to the
standard gradient:
\begin{equation}\label{eq:GradAtomMoire2}
    \nabla \bu_{\mathcal{M},1}(\bx) = \widehat{\nabla}_\bx \bu_1 (\bgamma_2(\bx)), \qquad \qquad \nabla \bu_{\mathcal{M},2}(\bx) = \widehat{\nabla}_\bx \bu_2 (\bgamma_1(\bx)).
\end{equation}
Let us finally reformulate the minimization problem~\eqref{def:CBModel} on the new domain $\Gamma_\mathcal{M}$:
\begin{equation}\label{def:CBModelMoire}
\begin{aligned}
    (\bu_{\mathcal{M}, 1}^*, \bu_{\mathcal{M}, 2}^*) \in \mathrm{arg\ min} \left \{ E_\mathcal{M}^\cb(\bu_{\mathcal{M}, 1}, \bu_{\mathcal{M}, 2}) \right. & + E_\mathcal{M}^\mathrm{misfit}(\bu_{\mathcal{M}, 1}, \bu_{\mathcal{M}, 2})  \\  \vert \ \bu_{\mathcal{M}, 1}, & \left.\bu_{\mathcal{M}, 2} \in H^1 \cap W^{1,\infty}(\Gamma_\mathcal{M}; \bbR^2) \right \}
\end{aligned}
\end{equation}
where
\begin{subequations}\label{eq:MoireCellEnergies}
    \begin{equation}
        E_\mathcal{M}^\cb(\bu_{\mathcal{M}, 1}, \bu_{\mathcal{M}, 2}) := \fint_{\Gamma_\mathcal{M}} W_1(\nabla \bu_{\mathcal{M}, 1}) + W_2(\nabla \bu_{\mathcal{M}, 2}),
    \end{equation}
    and
    \begin{equation}
    \begin{aligned}
        E_\mathcal{M}^\mathrm{misfit}(\bu_{\mathcal{M}, 1}, \bu_{\mathcal{M}, 2}) :=  \frac 1 2  \fint_{\Gamma_\mathcal{M}} &\Phi_{1+} \left (\bgamma_2(\bx) + \bu_{\mathcal{M}, 2}(\bx) - \bu_{\mathcal{M}, 1}(\bx) \right )  \\
         + &  \Phi_{2-}\left (\bgamma_1(\bx) + \bu_{\mathcal{M}, 1}(\bx) - \bu_{\mathcal{M}, 2}(\bx)  \right )d\bx.
    \end{aligned}
    \end{equation}
\end{subequations}
In the remainder of the section, we will drop the index $\mathcal{M}$ for simplicity whenever it is possible to do so without confusion.

\subsection{The space of admissible modulations}
 Now, it is easy to notice that the atomistic and continuum models are both formally translation invariant:
\[
    E^\mathrm{a}(\bu + \bx) = E^\mathrm{a}(\bu), \quad E^\cb(\bu + \bx) = E^\cb(\bu), \quad E^\mathrm{misfit}(\bu + \bx) = E^\mathrm{misfit}(\bu), \quad \forall \bx \in \bbR^2.
\]
In fact, the translation of each layer by independent vectors $\bx_1, \bx_2$ corresponds formally to choosing a different initial configuration, and in principle does not modify the energy landscape. Due to the approximation introduced by the interpolation scheme (Section~\ref{sec:atomistickinematics}), the misfit energy~\eqref{def:MisfitEnergy} is not exactly invariant under such independent translations of the layers; nevertheless these considerations justify formally the following additional constraint on modulation functions:
\begin{equation}\label{eq:AverageZero}
    \fint_{\Gamma_\mathcal{M}} \bu_j = \bzero, \qquad \text{for } j = 1,2.
\end{equation}
This leads us to introduce Sobolev spaces of periodic, null-average modulation functions for $n \geq 1$: on an arbitrary domain $\Gamma$:
\begin{equation}\label{def:HomFuncSpaces}
    \begin{aligned}
        W_\#^{n,p}(\Gamma) &:= \left \{\bu \in W^{n,p}(\Gamma, \bbR^2 \times \bbR^2) \ \text{such that}\ \fint_{\Gamma} \bu_1 = \fint_{\Gamma} \bu_2 = \bzero \right \}, \\
        W_\#^\infty(\Gamma) &:= \left \{\bv \in C^\infty(\Gamma, \bbR^2 \times \bbR^2) \ \text{such that}\ \fint_{\Gamma} \bv_1 = \fint_{\Gamma} \bv_2 = \bzero \right \}.
    \end{aligned}
\end{equation}
We typically omit the explicit dependency on $\Gamma$ when it is possible to do so without confusion, e.g. in this section $\Gamma \equiv \Gamma_\mathcal{M}$. The Hilbert space $W_\#^{n,2}$ is equipped with the norm:
\begin{equation}
    \Vert \bu \Vert_{n,2} := \left ( \sum_{k=0}^n \left \Vert \nabla^k \bu \right \Vert^2_{L^2} \right )^{1/2}, \qquad \text{where } \Vert \cdot  \Vert_{L^2} = \left ( \fint_{\Gamma} \Vert \cdot  \Vert^2 d \bx \right )^{1/2}.
\end{equation}
    Finally, to avoid interpenetration of matter, following~\cite{ortner2013justification} we assume that displacement gradients satisfy a uniform bound. In practice, we set a constant $0 < \kappa < 1$ and
\begin{equation}\label{def:InfBoundedGradients}
    K := \left \{ \bu \in W_\#^{1,\infty}(\Gamma_\mathcal{M}) \ \vert \ \Vert \nabla \bu_1 \Vert_{L^\infty} \leq \kappa,\ \Vert\nabla \bu_2 \Vert_{L^\infty} \leq \kappa\right \}.
\end{equation}
\subsubsection{Stability}

We assume that the lattices $\cR_j$ are stable, that is,
\begin{equation}\label{def:LayerStability}
    \nu := \min_{1 \leq j \leq p} \inf_{\substack{\bv \in C(X, \bbR^2) \\ \Vert \nabla \bv \Vert_{L^2} = 1}} \delta^2 E^a_j(0)[\bv, \bv] > 0,
\end{equation}
where we recall that the atomistic energy is defined by~\eqref{def:IntraAtomPotE}.
Physically, \eqref{def:LayerStability} ensures that small distortions of the lattice $\cR_j$ will increase the potential energy for layer $j$ in the stack.
The stability of the layers in the continuous framework can be then expressed by the following result, adapted from~\cite{ortner2013justification}, Proposition 5.1 and Lemma 5.2:
\begin{proposition}\label{prop:ContinuousLayerStability}
    Suppose the lattices $\cR_1, \cR_2$ are stable and let $\nu$ be defined by~\eqref{def:LayerStability}, then
    \begin{equation}\label{def:ContinuousLayerStability}
        \delta^2 E_\mathcal{M}^\cb(\mathbf{0}) [\bv, \bv] \geq \nu \Vert \nabla \bv \Vert^2_{L^2}, \qquad \forall \bv \in W_\#^{1,2}.
    \end{equation}
    Furthermore, there exists $\kappa_1 > 0$ such that, for $\kappa < \kappa_1$,
    \begin{equation}\label{def:ContinuousLayerStability2}
        \delta^2 E_\mathcal{M}^\cb(\bu) [\bv, \bv] \geq \frac 1 2 \nu \Vert \nabla \bv \Vert^2_{L^2}, \qquad \forall \bv \in W_\#^{1,2}, \quad \forall \bu \in K.
    \end{equation}
\end{proposition}
\begin{proof}
    As a first step, we note that the interpolant $\widehat{\bu}^c_j$ defined above~\eqref{def:InterpolantDisplacement2} allows us by straightforward calculations to write the following Birkhoff averaging formula. Fix $\bu_{\mathcal{M},j} \in K$, $j \in \{ 1,2 \}$,  and $f \in C \left (\{ \mathrm{F} \in R^{2 \times 2},\ \vert \mathrm{F} \vert \leq \kappa\} \right )$, then by Prop. 2.4,~\cite{Kubo2017}, for any $\omega \in \Omega$:
    \[
        \fint_{\Gamma_\mathcal{M}} f(\nabla \bu_{\mathcal{M},j}) = \fint_\Omega f(\nabla_\bx \widehat{\bu}^c_j) d\omega'   = \lim_{R \to \infty} \frac{1}{\vert B_R \vert} \int_{B_R} f\left(\nabla_\bx \widehat{\bu}_j(\mathtt{T}_{-\bx}\omega) \right ) d\bx.
    \]
    Now let us show that~\eqref{def:ContinuousLayerStability} holds. Fix $\bv_{\mathcal{M},j} \in C^\infty(\Gamma_\mathcal{M}, \bbR^2)$ and $\omega \in \Omega$. Let $\overline{\bv}_j \in C^\infty(\bbR^2, \bbR^2)$ defined by $\overline{\bv}_j(\bx) = {\bv}_j(\mathtt{T}_{-\bx}\omega)$, such that (see also~\cite{ortner2013justification}, Proposition 3.2):
    \begin{align*}
        \delta^2 E^\cb_{\mathcal{M},j} (\bzero)[\bv_{\mathcal{M},j}, \bv_{\mathcal{M},j}] &= \fint_{\Gamma_\mathcal{M}} \nabla^2 W_j \left ( \bzero \right ) [\nabla \bv_{\mathcal{M},j}, \nabla \bv_{\mathcal{M},j}] \\
        &= \lim_{R \to \infty} \frac{1}{\vert B_R \vert} \int_{B_R} \nabla^2 W_j\left(\bzero \right )[\nabla \overline{\bv}_j, \nabla \overline{\bv}_j] d\bx,
    \end{align*}
    where $E^\cb_{\mathcal{M},j} (\bv_{\mathcal{M},j})
     := \fint_{\Gamma_\mathcal{M}} W_j(\nabla \bv_{\mathcal{M},j}).$
    For $R > 0$, let $\phi_R \in C^\infty(\bbR^2, \bbR)$ be a smooth cutoff function such that $0 \leq \phi_R \leq 1$, $\phi_R(\bx) = 0$ for $\vert \bx \vert \geq R+1$, $\phi_R(\bx) = 1$ for $\vert \bx \vert \leq R$ and $\nabla \phi_R$ is uniformly bounded independently of $R$. By Lemma 5.2 in~\cite{ortner2013justification},
    \[
        \int_{\bbR^2} \nabla^2 W_j\left(\bzero \right )[\nabla (\phi_R \overline{\bv}_j),  \nabla (\phi_R\overline{\bv}_j)] d\bx \geq \nu \int_{\bbR^2} \left \vert \nabla (\phi_R \overline{\bv}_j) \right \vert \geq \nu \Vert \nabla \overline{\bv}_j \Vert^2_{L^2(B_R)}.
    \]
    Moreover
    \begin{align*}
        \int_{B_{R+1} \setminus B_R} \nabla^2 W_j\left(\bzero \right )[\nabla (\phi_R \overline{\bv}_j),  \nabla (\phi_R\overline{\bv}_j)] d\bx & \leq \Vert \nabla^2 W_j(\bzero) \Vert \int_{B_{R+1} \setminus B_R}\vert \nabla (\phi_R \overline{\bv}_j) \vert^2 \\
        & \leq C (R+1) \Vert \nabla \overline{\bv}_j \Vert^2_{L^2(B_{R+1} \setminus B_R)},
    \end{align*}
    where $C>0$ is a constant independent of $R$. We deduce in the thermodynamic limit:
    \begin{align*}
        \lim_{R \to \infty} \frac{1}{\vert B_R \vert} \int_{B_R} \nabla^2 W_j\left(\bzero \right )[\nabla \overline{\bv}_j,  \nabla \overline{\bv}_j] d\bx \geq \lim_{R \to \infty} \frac{\nu }{\vert B_R \vert} \Vert \nabla \overline{\bv}_j \Vert^2_{L^2(B_R)}.
    \end{align*}
    By the Birkhoff property above, this yields the configuration space estimate:
    \[
        \delta^2 E^\cb_j (\bzero)[\bv_{\mathcal{M},j}, \bv_{\mathcal{M},j}] \geq \nu \Vert \nabla \bv_{\mathcal{M},j} \Vert^2_{L^2}.
    \]
    Then~\eqref{def:ContinuousLayerStability} follows by density of $W_\#^\infty$ in $W_\#^{1,2}$. The proof of~\eqref{def:ContinuousLayerStability2} is analogous, using Proposition 5.1 in~\cite{ortner2013justification}.
\end{proof}

\subsubsection{Smoothness}
Using the same idea, we can also prove Lipschitz bounds on the variations of the intra-layer energies.
Following~\cite{ortner2013justification}, we define the intra-layer finite difference stencil space
\begin{equation}\label{def:IntraStencilSpace}
    \mathcal{D}_{j,\kappa} := \left \{ \mathbf{g} = (g_\rho)_{\rho \in \cB_j}\ \vert \ g_\rho \in \bbR^2, \Vert g \Vert_{\mathcal{D}_j} < \kappa \right \},
\end{equation}
equipped with the norm $\Vert \mathbf{g} \Vert_{\mathcal{D}_j} := \sup_{\rho \in \cB_j} \vert g_\rho \vert / \vert \rho \vert$.
We assume that $V^\mathrm{intra}_j \in C^k(\mathcal{D}_{j, \kappa})$, for some $k \geq 2$: the potential is smooth for injective deformations. Partial derivatives are denoted as
\[
    V^\mathrm{intra}_{j,\brho} (\bg) := \frac{\partial^i V_j^\mathrm{intra}(\bg)} {\partial g_{\rho_1} \ldots \partial g_{\rho_i}} \qquad \text{for } \mathbf{g} \in \mathcal{D}_{j, \kappa} \text{ and } \brho \in \left ( \cB_j \right ) ^i,
\]
and understood as multilinear forms acting on families of vectors $\bh = (h_1, \ldots, h_i)$, denoted as $V^\mathrm{intra}_{j,\brho} (\bg)[\bh] = V^\mathrm{intra}_{j,\rho_1 \ldots \rho_i}[h_1, \ldots, h_i]$, equipped with the norm
\begin{equation}\label{def:MultFormNorm}
    \Vert \ell \Vert := \sup_{\bh \in (\bbR^{2})^i, \vert h_1 \vert = \ldots = \vert h_i \vert = 1} \ell[\bh].
\end{equation}
Define quantities measuring the decay of the interatomic potentials $V^\mathrm{intra}_j$:
\[
    m_j(\brho) := \prod_{l = 1}^i \vert \rho_l \vert \sup_{\bg \in \mathcal{D}_{j, \kappa}} \Vert V_{j, \brho}^\mathrm{intra}(\bg) \Vert \qquad \text{for } \brho \in (\cB_j)^i,\quad 1 \leq i \leq k.
\]
Then, as a consequence of Prop. 3.2 in~\cite{ortner2013justification}:
\begin{proposition}\label{prop:VariationsBounds}
    Suppose the atomic potentials satisfy the decay bounds
    \begin{equation}\label{def:AtomisticPotMoments}
    M^{(i)} := \max_{1 \leq j \leq p}  \sum_{\brho \in (\cB_j)^i} m_j(\brho) < \infty, \qquad \text{for } 1 \leq i \leq k.
\end{equation}
Then for $k \leq 5$, $\sum_{i = 1}^k \frac{1}{p_i} = 1$, we have
    \begin{equation}\label{def:ContinuousPotMoments}
        \delta^k E^\cb_\mathcal{M}(\bu)[\bv_1, \dots, \bv_k] \leq M^{(k)} \prod_{i=1}^k \Vert \nabla \bv_i \Vert_{L^{p_i}}.
    \end{equation}
\end{proposition}

\subsection{Existence of local energy minimizers}
Our goal in this section is to prove the following existence theorem for the energy minimizer in the bilayer case:
\begin{theorem}\label{th:Existence}
    Suppose the lattices $\cR_1, \cR_2$ are stable in the sense of~\eqref{def:LayerStability}, then when the site-energies defined earlier~\eqref{def:GSFE} are such that $\Vert \Phi_{1+} \Vert_{C^4(\Gamma_2)}$ and $\Vert \Phi_{2-} \Vert_{C^4(\Gamma_1)}$  are small enough (weak inter-layer coupling regime),
there exists a solution $\bu^* \in W_\#^{3,2}$ of the Euler-Lagrange equation,
    \begin{equation}\label{eq:EulerLagrange}
        \delta E^\cb_\mathcal{M}(\bu^*)\left [ \bv\right ] + \delta E^\mathrm{misfit}_\mathcal{M}(\bu^*)\left [ \bv\right ] = 0,  \qquad \forall \bv \in W_\#^{3,2},
    \end{equation}
    or equivalently $\bu^*$ is a solution of two coupled partial differential equations on $\Gamma_\mathcal{M}$:
    \begin{equation}\label{eq:EulerLagrangeStrong}
        \begin{aligned}
            & \left \{ \begin{aligned}
             -\mathbf{div} \left ( \nabla W_1(\nabla \bu_1) \right ) &= \mathbf{f}(\bx, \bu_1 - \bu_2), \\
             -\mathbf{div} \left ( \nabla W_2(\nabla \bu_2) \right ) &= -\mathbf{f}(\bx, \bu_1 - \bu_2),
            \end{aligned} \right. \\
            & \quad \text{where} \quad \mathbf{f}(\bx, \bv) := \frac 12  \left (\nabla \Phi_{1+} \right ) \left [ \bgamma_2(\bx) - \bv \right ] - \frac 12 \left (\nabla \Phi_{2-} \right ) \left [ \bgamma_1(\bx) + \bv \right ] .
        \end{aligned}
    \end{equation}
    Furthermore this solution is stable in the sense that there exists $c_0 > 0$ such that:
    \begin{equation}\label{eq:SolStability}
        \delta^2 E^\cb_\mathcal{M}(\bu^*)\left [ \bv,\bv\right ] + \delta^2 E^\mathrm{misfit}_\mathcal{M}(\bu^*)\left [ \bv,\bv\right ] \geq c_0 \Vert \bv \Vert_{1,2}^2, \qquad \forall \bv \in W_\#^{1,2}.
    \end{equation}
\end{theorem}

\begin{remark}
     The multi-well nature of the nonlinear forcing term in the right-hand side of the Euler-Lagrange equation~\eqref{eq:EulerLagrangeStrong} is reminiscent of the Ginzburg-Landau equation, as noted in~\cite{Espanol2d}.
\end{remark}

\begin{remark}
    The existence of \textit{global} energy minimizers can also be shown by standard methods
    exploiting the convexity or quasi-convexity properties of typical linear or nonlinear elastic functionals approximating the CB energy $E^\cb$ with respect to the gradient $\nabla_\bx \bu$ in the space of null-average modulations
    $W_\#^{1,p}$.
    The proofs of Theorem 9.5-2 in~\cite{ciarlet2013linear} in the linear case, or Theorem 8.31 in~\cite{dacorogna2007direct} for the non-linear case, apply almost verbatim with slight modifications to account for the periodic boundary conditions.
\end{remark}
\begin{proof}
    Our proof relies on an application of a quantitative version of the inverse function theorem, for example Lemma 2.2 in~\cite{ortner2011priori} which we recall here for completeness:
    \begin{lemma}\label{lem:InverseFunctionTheorem}
        Let $\cX$, $\cY$ be Banach spaces, $\cA$ an open subset of $\cX$, and let $\mathcal{F} : \cA \to \cY$ be Fréchet differentiable. Suppose that $x_0 \in \cA$ satisfies the conditions:
        \[
            \Vert \mathcal{F}(x_0) \Vert_\cY \leq \eta, \quad \Vert \delta \mathcal{F}(x_0)^{-1} \Vert_{L(\cY, \cX)} \leq \sigma, \qquad \overline{B_\cX(x_0, 2\eta\sigma) \subset \cA},
        \]
        $\Vert \delta \mathcal{F}(x_1) - \delta \mathcal{F}(x_2) \Vert_{L(\cX,\cY)} \leq L \Vert x_1 - x_2 \Vert_\cX \quad$ for $\quad \Vert x_j - x_0 \Vert_\cX \leq 2 \eta \sigma,$ \\
        and $2 L \sigma^2 \eta < 1$, \\
        then there exists a unique $x \in \cX$ such that $\mathcal{F}(x) = 0$ and $\Vert x - x_0 \Vert_\cX \leq 2 \eta \sigma$.
    \end{lemma}
    Let us set $\cX := W_\#^{3,2}$, $\cY :=  W_\#^{1,2}$, $x_0 = \bzero$ and $\cA = B_\cX(\mathbf{0}, \kappa')$ where $\kappa'$ is chosen to be small enough that $\cA \subset K$ thanks to the Sobolev embedding $W_\#^{3,2} \subset W_\#^{1,\infty}$.
    Define the mapping $\mathcal{F}: \cA \to \cY$ defined by
    \[
        \mathcal{F}(\bu) := \delta E^\cb_\mathcal{M}(\bu) + \delta E^\mathrm{misfit}_\mathcal{M}(\bu).
    \]
    We compute explicitly from~\eqref{eq:MoireCellEnergies}, with $\mathbf{f}$ defined as in~\eqref{eq:EulerLagrangeStrong}:
    \begin{equation}
        \delta E^\cb_\mathcal{M}(\bu) =
        \begin{bmatrix}
            - \mathbf{div} \left ( \nabla W_1 \left (\nabla \bu_1 \right )  \right ) \\
            - \mathbf{div} \left ( \nabla W_2 \left (\nabla \bu_2 \right )  \right )
        \end{bmatrix}, \qquad
        \delta E^\mathrm{misfit}_\mathcal{M}(\bu) =
             \begin{bmatrix} - \mathbf{f} \left ( \bx, \bu_1 - \bu_2 \right ) \\ \mathbf{f} \left ( \bx, \bu_1 - \bu_2\right) \end{bmatrix}.
    \end{equation}

\paragraph{Residual} In particular, we have
    \begin{equation}\label{eq:invtheorem1}
        \Vert \mathcal{F}(\bzero) \Vert_\cY = \Vert \delta E^\mathrm{misfit}_\mathcal{M}(\bzero) \Vert_{1,2} = \sqrt{2} \Vert \mathbf{f} (\, \cdot\, , \bzero) \Vert_{1,2} \leq \eta,
    \end{equation}
    where $\eta = c \sum_{j = 1}^2 \left ( \Vert \nabla \Phi_{j,\pm} \Vert_{\infty} + \Vert \mE_2 - \mE_1 \Vert  \Vert \mE_j^{-1} \Vert \Vert \nabla^2 \Phi_{j,\pm} \Vert_{\infty} \right )$ and $c$ is a generic constant independent of $\cR_j$ and $\Phi_{j\pm}$.

\paragraph{Stability} Let us compute for $\bv \in \cX$:
    \[
        \delta^2 E^\mathrm{misfit}_\mathcal{M}(\bzero)[\bv, \bv] = \frac 1 2 \fint_{\Gamma_\mathcal{M}} \left ((\nabla^2 \Phi_{1+})\circ \bgamma_2 + (\nabla^2 \Phi_{2-}) \circ \bgamma_1 \right )\left [ \bv_2 - \bv_1, \bv_2 - \bv_1 \right ] \mathrm{d}\bx ,
    \]
    so that, using Proposition~\ref{prop:ContinuousLayerStability}:
    \begin{equation}\label{eq:stability}
        \begin{aligned}
        \delta^2 E^\mathrm{CB}_\mathcal{M}(\bzero)[\bv, \bv] + \delta^2 E^\mathrm{misfit}_\mathcal{M}(\bzero)[\bv, \bv] & \geq \nu \Vert \nabla \bv \Vert^2_{L^2} - \left ( \sum_{j = 1}^2 \Vert \nabla^2 \Phi_{j,\pm} \Vert_{\infty} \right )  \Vert \bv \Vert^2_{L^2} \\
        & \geq \nu / 2 C^2_P \Vert \bv \Vert^2_{1,2},
        \end{aligned}
    \end{equation}
    assuming that $\sum_{j = 1}^2 \Vert \nabla^2 \Phi_{j,\pm} \Vert_{\infty} <  \nu / 2 C^2_P$,
    where $C_P$ is the constant in the Poincaré-Wirtinger estimate:
    \begin{equation}\label{eq:PoincareWirtinger}
        \left \Vert v - \langle v \rangle\right \Vert_{1,2} \leq C_P \Vert \nabla v \Vert_{L^2(\Gamma_M)} \qquad \text{where } \langle v \rangle = \fint_{\Gamma_M} v d\bx.
    \end{equation}
    Fix $\bv \in \cY$.
    By the Lax-Milgram theorem, there exists a unique solution $\bu \in \cY$ to the problem $\delta \mathcal{F}(\bzero)[\bu] = \bv$ with $\Vert \bu \Vert_{1,2} \leq 2\nu^{-1}C^2_P \Vert \bv \Vert_{1,2}$.
    Let $\bg = - \delta^2 E^\mathrm{misfit}_\mathcal{M}(\bzero)[\bu,\,\cdot\,]$, we compute
    \[
    \Vert \bg \Vert_{1,2} \leq c \left ( \sum_{j = 1}^2  \Vert \nabla^2 \Phi_{j,\pm} \Vert_{\infty} + \Vert \mE_2 - \mE_1 \Vert  \Vert \mE_j^{-1} \Vert \Vert \nabla^3 \Phi_{j,\pm} \Vert_{\infty}  \right ) \Vert \bu \Vert_{1,2},
    \]
    where $c$ is another generic constant.
    Then $\bu$ is a solution of the elliptic problem with constant coefficients: $\delta^2 E^\mathrm{CB}_\mathcal{M}(\bzero)[\bu, \cdot] = \bv + \bg$.
    Elliptic regularity (see e.g.~\cite{Evans2010}, section 6.3.1) ensures that $\bu$ belongs to $W_\#^{3,2}$
    and one can further show that $\Vert \bu \Vert_{3,2} \leq  C \Vert \bv + \bg \Vert_{1,2} $ for some constant $C>0$ independent of $\Phi_{j\pm}$.
    This shows that $\delta \mathcal{F}(\bzero)$ is a bijection from $W_\#^{3,2}$ to $W_\#^{1,2}$ with the stability estimate:
    \begin{equation}\label{eq:invtheorem2}
    \Vert \delta\mathcal{F}(\bzero)^{-1} \Vert_{L(\cY, \cX)} \leq \sigma, \qquad \text{where } \sigma = C \left (1 + C' \sum_{j = 1}^2 \Vert \Phi_{j \pm} \Vert_{C^3} \right ) .
    \end{equation}

\paragraph{Lipschitz estimate}
    Using the Sobolev inequalities $\Vert \bu \Vert_{1,\infty} \lesssim \Vert \bu \Vert_{2,4} \lesssim \Vert \bu \Vert_{3,2}$ and estimate~\eqref{def:ContinuousPotMoments}, a lengthy but straightforward computation for arbitrary $\bu, \bu', \bv \in W_\#^{3,2}$ leads to:
    \begin{align*}
        \Big \Vert \delta^2 E^\cb_\mathcal{M}(\bu)[\bv, \cdot] & - \delta^2 E^\cb_\mathcal{M}(\bu')[\bv, \cdot]  \Big \Vert_{1,2} \\
        &= \left \Vert \begin{bmatrix}
            - \mathbf{div} \left ( \nabla^2 W_1(\nabla \bu_1) : \nabla \bv_1 - \nabla^2 W_1(\nabla \bu'_1) : \nabla \bv_1 \right ) \\
            - \mathbf{div} \left ( \nabla^2 W_2(\nabla \bu_2) : \nabla \bv_2 - \nabla^2 W_2(\nabla \bu'_2) : \nabla \bv_2 \right )
        \end{bmatrix}   \right \Vert_{1,2}\\
        & \leq C \sum_{k = 3}^5 M^{(k)} \Vert \bu - \bu'\Vert_{3,2} \Vert \bv\Vert_{3,2},
    \end{align*}
    where $C$ is a generic constant, independent of $\cR_j$ and $\Phi_{j\pm}$.
    Similarly, we compute
    \begin{align*}
        \Big \Vert \delta^2 E^\mathrm{misfit}_\mathcal{M}(\bu)&[\bv, \cdot]  - \delta^2 E^\mathrm{misfit}_\mathcal{M}(\bu')[\bv, \cdot]  \Big \Vert_{1,2} \\
         & \leq \frac 12 \sum_{j=1}^2  \left \Vert
        \begin{bmatrix}
            -\left ( \nabla_\bv \mathbf{f}_j \left ( \bx, \bu_1 - \bu_2 \right ) - \nabla_\bv \mathbf{f}_j \left ( \bx, \bu_1' - \bu_2' \right ) \right ) \cdot (\bv_1 - \bv_2) \\
            \left ( \nabla_\bv \mathbf{f}_j \left ( \bx, \bu_1 - \bu_2 \right ) - \nabla_\bv \mathbf{f}_j \left ( \bx, \bu_1' - \bu_2' \right ) \right ) \cdot (\bv_1 - \bv_2)
        \end{bmatrix}   \right \Vert_{1,2}\\
        & \leq C \Vert \bu - \bu' \Vert_{W^{1,\infty}} \Vert \bv \Vert_{W^{1,2}} \leq  C' \Vert \bu - \bu'\Vert_{3,2} \Vert \bv\Vert_{3,2},
    \end{align*}
    with $C' = c\ \sum_{j=1}^2 \left (  \kappa \Vert \nabla^4 \Phi_{j\pm} \Vert_{\infty} + \Vert \mE_2 - \mE_1 \Vert \Vert \mE_j^{-1} \Vert \Vert \nabla^4 \Phi_{j\pm} \Vert_{\infty} +  \Vert \nabla^3 \Phi_{j\pm} \Vert_{\infty} \right )$
    and $c$ a generic constant. This shows that $\mathcal{F}$ is continuously differentiable on $\Xi$:
    \begin{align*}
        \left \Vert \delta \mathcal{F} (\bu) - \delta \mathcal{F} (\bu') \right \Vert_{L(\cX, \cY)} &\leq L \Vert \bu - \bu'\Vert_{3,2} \Vert \bv\Vert_{3,2} \quad \text{with} \quad L := \left ( C \sum_{k = 3}^5 M^{(k)} + C'\right ).
    \end{align*}
    We find finally that $2 \eta \sigma < \kappa'$ and $2L\sigma\eta^2 < 1$ for $\sum_{j = 1}^2 \Vert \Phi_{j, \pm} \Vert_{C^4}$ small enough.
    Then by the quantitative implicit function theorem, Lemma~\ref{lem:InverseFunctionTheorem}, there exists a unique solution $\bu^*$ to the Euler-Lagrange equation, $\mathcal{F}(\bu^*) = \bzero$, such that $\Vert \bu^* \Vert_{3,2} < 2 \eta \sigma$.

    It remains to show that this solution is stable.

    We assume that $\Vert \Phi_{j,\pm} \Vert_{C^2}$ are small enough such that we may choose $c_0 > 0$ satisfying the condition:
    \begin{equation}\label{eq:CoercivityConstant}
        0 < c_0 < \frac{\nu}{2 C_P} - \sum_{j = 1}^2 \Vert \Phi_{j,\pm} \Vert_{C^2}.
    \end{equation}
    Then, computing as above in~\eqref{eq:stability}, Proposition~\ref{prop:ContinuousLayerStability} ensures that:
    \begin{equation}\label{eq:Coercivity}
        \delta^2 E^\cb_\mathcal{M}(\bu^*)\left [ \bv,\bv\right ] + \delta^2 E^\mathrm{misfit}_\mathcal{M}(\bu^*) \left [ \bv,\bv\right ] \geq c_0 \Vert \bv \Vert^2_{L^2}.
    \end{equation}
    Then $\bu^*$ is a stable solution in the sense of~\eqref{eq:SolStability}.
\end{proof}

\subsection{Asymptotic analysis}
Let us describe more precisely, in a semilinear setting, the parameter range where we can ensure a unique solution such as given by Theorem~\ref{th:Existence} exists.
Equation~\eqref{eq:CoercivityConstant} gives a hint:
when $\Vert \mE_1 - \mE_2 \Vert \to 0$, the diameter of the moiré cell~\eqref{def:MoireCell}, and consequently the constant $C_P$ in the Poincaré-Wirtinger estimate~\cite{Evans2010}, grows to infinity roughly as $\Vert (\mE_1 - \mE_2)^{-1}\Vert$ until the right-hand side in~\eqref{eq:CoercivityConstant} becomes negative.
In particular, the physically relevant inter-layer interactions, albeit weak, may be too large for Theorem~\ref{th:Existence} to apply at small angles.
One possibility is for example that for large-scale moiré patterns, a bifurcation occurs where two or more stable solution branches exist~\cite{Dai2016}.

To simplify the analysis, we will assume that both layers are made of the same material, and the stresses inside each layer are modeled by isotropic linear elasticity,
i.e., the intra-layer energy density functionals are the same $W := W_1 = W_2$ which takes the form
\begin{equation}\label{def:LinearElasticFunctional}
    W(\bu) = \frac \lambda 2  \mathrm{div}(\bu)^2 + \mu \varepsilon(\bu) : \varepsilon(\bu) \qquad \text{where } \varepsilon(\bu) = \frac 12 (\nabla \bu + \nabla \bu^T),
\end{equation}
where $\lambda \geq 0$ and $\mu > 0$ are the Lamé parameters of the material.
In addition, the lattices basis of both layers are identical, and thus share the same lattice basis $\mE$ simply twisted by a relative angle $\theta > 0$:
\begin{equation}\label{eq:TwistLattices}
    \mE_1 = \mR_{-\theta/2}\mE, \qquad \mE_2 = \mR_{\theta/2} \mE, \qquad \qquad \text{where } \mR_\theta = \begin{bmatrix} \cos(\theta) & -\sin(\theta) \\ \sin(\theta) & \cos(\theta) \end{bmatrix}.
\end{equation}
A direct computation then shows that for varying $\theta$ the moiré cell~\eqref{def:MoireCell} is simply a rescaled copy of a new reference cell $\Gamma_0$:
\begin{equation}\label{eq:TwistMoire}
    \Gamma_\mathcal{M}(\theta) = \frac{1}{2\sin(\theta/2)} \Gamma_0,
    \quad \text{ where }
    \begin{cases}
        \cR_0 := \mE_0 \bbZ^2, \\
        \Gamma_0 := \bbR^2 / \cR_0,
    \end{cases}
    \text{with } \mE_0 := \mR_{-\pi/2} \mE.
\end{equation}
We rescale~\eqref{def:MoireMappings} to define consistent disregistry mappings on the reference cell $\Gamma_0:$
\begin{equation}\label{def:MoireRefMappings}
    \bgamma_{1,\theta}: \left \{ \begin{aligned}
        \Gamma_0 &\to \Gamma_1, \\ \bx &\mapsto -\mR_{-\frac{\pi+\theta}{2}} \bx,
    \end{aligned} \right. \qquad \qquad
    \bgamma_{2,\theta}: \left \{ \begin{aligned}
        \Gamma_0 &\to \Gamma_2, \\ \bx &\mapsto \mR_{\frac{\theta-\pi}{2}} \bx.
    \end{aligned} \right.
\end{equation}
Also, thanks to the symmetry relation~\eqref{eq:GSFESym} the generalized stacking fault energy can be transposed as a single functional:
\begin{equation}
    \Phi_0(\bx) := \Phi_{1+}\left (\bgamma_{2,\theta}(\bx) \right ) = \Phi_{2-}\left (\bgamma_{1,\theta}(\bx) \right ).
\end{equation}

The Euler-Lagrange equation~\eqref{eq:EulerLagrangeStrong} for extremal points of the energy functional for any angle $0 < \theta < \pi$ can be recast as a problem on the reference cell $\Gamma_0$:
\begin{equation}\label{eq:EulerLagrangeStrongRefSystem}
     \left \{ \begin{aligned}
         -\mathbf{div} \left ( \lambda \mathrm{div}(\bu_1) \mI + 2 \mu \varepsilon(\bu_1) \right ) &= \frac{1}{4\sin^2(\theta/2)} \mathbf{F}_\theta(\bx, \bu_1 - \bu_2), \\
         -\mathbf{div} \left ( \lambda \mathrm{div}(\bu_2) \mI + 2 \mu \varepsilon(\bu_2) \right ) &= \frac{- 1}{4\sin^2(\theta/2)} \mathbf{F}_\theta(\bx, \bu_1 - \bu_2),
     \end{aligned} \right.
\end{equation}
where
\begin{equation}\label{def:MoireRefGSFEForces}
    \mathbf{F}_\theta(\bx, \bv) 
        := \frac{1}{2} \nabla \Phi_0(\bx - \mR_{\frac{\pi - \theta}{2}} \bv ) \cdot \mR_{\frac{\pi - \theta}{2}} + \frac{1}{2} \nabla \Phi_0(\bx - \mR_{\frac{\pi + \theta}{2}} \bv ) \cdot \mR_{\frac{\pi + \theta}{2}}. 
\end{equation}
\begin{remark}
    Local strains are usually small in the physically relevant regime of weak inter-layer interactions and moire-scale modulations, and it is reasonable as a first approximation to model the stresses inside each layer using isotropic linear elasticity. At the same time, the misfit energy exhibits large variations for atomic-scale displacements and hence forces $\mathbf{F}_\theta$ due to the inter-layer misfit cannot be linearized.
\end{remark}
It can be easily be shown that any solution of the system~\eqref{eq:EulerLagrangeStrongRefSystem} satisfies $\bu_1 = - \bu_2$. Then the model~\eqref{eq:EulerLagrangeStrongRefSystem} reduces simply to the single semilinear equation:
\begin{equation}\label{eq:EulerLagrangeStrongRef}
     -\mathbf{div} \left ( \lambda \mathrm{div}(\bu) \mI + 2 \mu \varepsilon(\bu) \right ) = \frac{1}{4\sin^2(\theta/2)} \mathbf{F}_\theta(\bx, 2\bu),  \quad \text{where } \bu := \bu_1 = -\bu_2.
\end{equation}
\begin{remark}
    At small angles, a Taylor expansion of expression~\eqref{def:MoireRefGSFEForces} indicates that $F_\theta$ depends only very weakly on $\theta$:
    \[
        \vert F_\theta(\bx, \bv) - F_0(\bx, \bv) \vert = \mathcal{O}(\theta^2 (1+\vert \bv \vert)).
    \]
    Therefore the main parameters governing the relaxation are the profile of the GSFE functional, mainly the location of the minima, saddle points and maxima which all depends mostly on the symmetries of the material, and the ratio of the GSFE amplitude to the square of the twist angle.
\end{remark}
\begin{remark}
    The natural rescaling of the gradient by the inverse moiré length $(2 \sin(\theta/2))^{-1}$ in~\eqref{eq:EulerLagrangeStrongRef} is reminiscent of the rescaling by $\varepsilon := \sigma/L$ in the continuum Euler-Lagrange equations in~\cite{Espanol2d} where $\sigma$ is the inter-layer distance, and $L$ the lateral size of a domain which is constrained to be of the same order as the moiré pattern.
\end{remark}

This explicit rescaling enables us to study explicitly the interplay between twist angle and strength of the inter-layer coupling:
\begin{theorem}
    Assume the functional $\Phi_0$ is twice continuously differentiable.
    \begin{enumerate}[label=(\alph*)]
        \item For any $\theta \neq n \pi$ for any integer $n$ there exists solutions to the Euler-Lagrange equation~\eqref{eq:EulerLagrangeStrongRef} in $W^{1,2}_\#(\Gamma_0)$, which satisfy
            \begin{equation}\label{eq:H1Bound}
                \Vert \bu^* \Vert_{1,2} \leq  \frac{1}{\mu C^1_\mE} \frac{\Vert \nabla \Phi_0 \Vert_\infty}{4\sin^2(\theta/2)} ,
            \end{equation}
            where $C_{\mE}^1$ is a fully identifiable constant depending on the lattice basis $\mE$ only.
        \item The solution is unique if
            \begin{equation}\label{eq:BanachFixedPoint}
                \frac{\Vert \nabla^2 \Phi_0 \Vert_\infty}{\theta^2}  <  \frac{\mu C^0_\mE}{2},
            \end{equation}
            where $C_\mE^0 > 2 C_\mE^1$ is another fully identifiable constant depending on the lattice basis $\mE$ only.
        \item If in addition $\Phi_0$ is smooth, any solution is also smooth and its derivatives for any $k = 2, 3, \dots$ satisfy the bounds
            \begin{equation}\label{eq:HkBound}
                \Vert \nabla^{k} \bu^* \Vert_{L^2} \leq \frac{C_k}{\sin^{2k}(\theta/2) },
            \end{equation}
            where the constant $C_k$ depends only on $\mE$, $\mu$ and $\Vert \Phi_0 \Vert_{C^{k+1}(\Gamma_0)}$.
    \end{enumerate}
\end{theorem}
\begin{proof}
    Given any $\bg \in L^2(\Gamma_0; \bbR^2)$, consider the problem: find $\bu$ such that
    \[
        -\mathbf{div} \left ( \lambda \mathrm{div}(\bu) \mI +  2 \mu \varepsilon(\bu) \right ) = \bg.
    \]
    The corresponding variational formulation writes:
    \begin{equation}\label{eq:VarForm}
         \int_{\Gamma_0} \lambda \mathrm{div}(\bu)\mathrm{div}(\bv) + 2 \mu \varepsilon(\bu) : \varepsilon(\bv) = \int_{\Gamma_0} \bg \cdot \bv,  \qquad \forall \bv \in W^{1,2}(\Gamma_0, \bbR^2).
\end{equation}
To establish the coercivity of the bilinear form on the left-hand side of~\eqref{eq:VarForm}, we decompose periodic, zero-average functions over $\Gamma_0$ as the Fourier series:
\[
    \bv(\bx) = \sum_{\bk \in \cR_0^*\setminus\{\bzero\}} \bv_\bk e^{2 i \pi \bk \cdot \bx}, \qquad \text{such that } \Vert \bv \Vert_{1,2}^2 = \sum_{\bk \in \cR_0^*} (1 + \vert 2\pi\bk \vert^2) \vert \bv_\bk \vert^2.
\]
Now we estimate:
\begin{align*}
    \int_{\Gamma_0} \lambda \mathrm{div}(\bv)^2 + 2 \mu \varepsilon(\bv) : \varepsilon(\bv) & \geq 2 \mu \int_{\Gamma_0} \varepsilon(\bv) : \varepsilon(\bv) \\
    & \geq 2 \mu  \sum_{\bk \in \cR_0^*\setminus\{\bzero\}}  (2\pi)^2 \left \vert \frac{\bk \otimes \bv_\bk + \bv_\bk \otimes \bk}{2} \right \vert^2 \\
    & \geq \mu  \sum_{\bk \in \cR_0^*\setminus\{\bzero\}} \vert 2 \pi \bk \vert^2 \vert \bv_\bk \vert^2.
\end{align*}
This implies the coercivity estimates:
\begin{equation}\label{eq:Korn}
    \begin{aligned}
        \int_{\Gamma_0} \lambda \mathrm{div}(\bv)^2 + 2 \mu \varepsilon(\bv) : \varepsilon(\bv) \geq \mu \Vert \nabla \bv \Vert_{L^2}^2, \\
        \int_{\Gamma_0} \lambda \mathrm{div}(\bv)^2 + 2 \mu \varepsilon(\bv) : \varepsilon(\bv) \geq \mu C^0_\mE \Vert \bv \Vert_{L^2}^2, \\ \int_{\Gamma_0} \lambda \mathrm{div}(\bv)^2 + 2 \mu \varepsilon(\bv) : \varepsilon(\bv) \geq \mu C^1_\mE \Vert \bv \Vert_{1,2}^2,
    \end{aligned}
\end{equation}
where $C^0_\mE = \min \left \{ \vert 2 \pi  \bk \vert^2 \ \big \vert \ \bk \in \cR^*_0 \setminus \{ \bzero \} \right \}$ and $C^1_\mE = \frac{C^0_\mE}{1+C^0_\mE} $. The Lax-Milgram lemma then ensures that there exists unique solutions $ \mathcal{G}(\bg) := \bu \in W_\#^{1,2}$ to problem~\eqref{eq:VarForm}, such that $\mathcal{G}$ is a continuous linear mapping $L^2 \to W_\#^{1,2}$ satisfying the bounds
\begin{equation*}
\Vert \mathcal{G}(\bg) \Vert_{L^2} \leq \frac{1}{\mu C^0_\mE} \Vert \bg \Vert_{L^2},\qquad
\Vert \mathcal{G}(\bg) \Vert_{W_\#^{1,2}} \leq \frac{1}{\mu C^1_\mE} \Vert \bg \Vert_{L^2}.
\end{equation*}
We also introduce the mapping $\mathcal{F}: L^2_\# \to L^2$ with $\mathcal{F}(\bu) = \frac{1}{4\sin^2(\theta/2)} \bF_\theta(\cdot, 2 \bu(\cdot))$, which is bounded: $\Vert \mathcal{F}(\bu) \Vert_{L^2} \leq \frac{1}{4\sin^2(\theta/2)} \Vert \nabla \Phi_0 \Vert_\infty$, as well as Lipschitz continuous:
\[
    \Vert \mathcal{F}(\bu) - \mathcal{F}(\bu') \Vert_{L^2} \leq  \frac{2}{4\sin^2(\theta/2)} \Vert \nabla^2 \Phi_0 \Vert_{\infty} \Vert \bu - \bu' \Vert_{L^2_\#}.
\]
Now, we note that the mapping $\mathcal{H} = \mathcal{G} \circ \mathcal{F}: L_\#^2 \to W_\#^{1,2}$ is bounded,
\begin{equation}\label{compact}
\Vert \mathcal{H}(\bu)\Vert_{L^2} \leq \frac{1}{4\sin^2(\theta/2)}  \frac{\Vert \nabla \Phi_0 \Vert_\infty}{\mu C^0_\mE},\quad
\Vert \mathcal{H}(\bu)\Vert_{W_\#^{1,2}} \leq \frac{1}{4\sin^2(\theta/2)}  \frac{\Vert \nabla \Phi_0 \Vert_\infty}{\mu C^1_\mE},
\end{equation}
and thus maps continuously the following ball of $L^2_\#(\Gamma_0, \bbR^2)$:
\[
    B = \left \{ \bu \in L^2_\# \ \big \vert\  \Vert \bu \Vert_{L^2} \leq  \frac{1}{4\sin^2(\theta/2)}  \frac{\Vert \nabla \Phi_0 \Vert_\infty}{\mu C^0_\mE}\right \}.
\]
Then, because the injection $W^{1,2}_\# \hookrightarrow L^2_\#$ is compact, we can apply Schauder's fixed point theorem (see e.g. Theorem 9.12-1 in~\cite{ciarlet2013linear}) to obtain the existence of at least one fixed point of the mapping $\mathcal{H}$, which is a solution to~\eqref{eq:EulerLagrangeStrongRef}. The bound~\eqref{eq:H1Bound} of (a) follows from \eqref{compact}.

Furthermore, $\mathcal{H}$ satisfies the Lipschitz bound:
\[
    \Vert \mathcal{H} (\bu) - \mathcal{H} (\bu') \Vert_{L^2} \leq \frac{1}{4\sin^2(\theta/2)}  \frac{2\Vert \nabla^2 \Phi_0 \Vert_\infty}{\mu C^1_\mE}  \Vert \bu - \bu' \Vert_{L^2}.
\]
Thus we may apply the Banach fixed point theorem (see e.g. Theorem 3.7-1 in~\cite{ciarlet2013linear}) to show there exists a unique fixed point to $\mathcal{H}$ under the condition:
\[
    \frac{1}{4\sin^2(\theta/2)}  < \frac{\mu C^1_\mE}{2\Vert \nabla^2 \Phi_0 \Vert_\infty}.
\]
Since $4 \sin^2(\theta/2) \leq \theta^2 $, this proves (b).

Finally, assume that $\Phi_0$ is smooth and $\bu^* \in W^{1,2}_\#$ is a solution of the Euler-Lagrange system~\eqref{eq:EulerLagrangeStrongRef}. The finite difference technique for elliptic regularity (see e.g.~\cite{Evans2010}, Section 6.3.1) ensures that $\bu^* \in W_\#^{k,2}$ for any $k \geq 2$. Let us show the bound~\eqref{eq:HkBound} by recursion.
Set:
\[
    \bg = \bF_\theta(\bx, 2 \bu^*).
\]
For $k = 2$, we test~\eqref{eq:VarForm} with $\bv' = \partial_\gamma \bv$ for some $\bv \in C^\infty(\Gamma_0, \bbR^2)$ and $\ell = 1,2$. Then
\begin{align*}
    \frac{1}{4\sin^2(\theta/2)}  \int_{\Gamma_0} \bg \cdot \partial_\ell \bv & = \int_{\Gamma_0} \lambda \mathrm{div}(\bu)\mathrm{div}(\partial_\ell \bv) + 2 \mu \varepsilon(\bu) : \varepsilon(\partial_\ell \bv) \\
    & = - \int_{\Gamma_0} \lambda \mathrm{div}(\partial_\ell  \bu)\mathrm{div}( \bv) + 2 \mu \varepsilon(\partial_\ell  \bu) : \varepsilon(\bv),
\end{align*}
which implies that $\mu \Vert \nabla \partial_\ell \bu \Vert_{L^2} \leq \frac{1}{4\sin^2(\theta/2)}  \Vert \bg \Vert_{L^2} \leq \frac{1}{4\sin^2(\theta/2)}  \Vert \nabla \Phi_0 \Vert_{L^\infty}$ and thanks to~\eqref{eq:H1Bound}:
\[
     \Vert \bu \Vert_{2,2} \leq \frac{C_2}{4 \sin^2(\theta/2)},
\]
    where the constant $C_2$ depends only on $\mE$, $\mu$ and $\Vert \Phi_0 \Vert_{C^{1}(\Gamma_0)}$. Next, for $k \geq 2$ we use the test function $\Delta^k \bu^*$ and the Fourier transform:
    \[
        \begin{aligned}
            \sum_{\bk \in \cR_0^*\setminus\{\bzero\}}  (2\pi)^2 \vert 2\pi\bk \vert^{2k} \left ( 2 \mu \left \vert \frac{\bk \otimes \bu^*_\bk + \bu^*_\bk \otimes \bk}{2} \right \vert^2 + \lambda \vert  \bk \cdot \bu_\bk^* \vert^2 \right ) \\
            = \frac{1}{4\sin^2(\theta/2)} \sum_{\bk \in \cR_0^*\setminus\{\bzero\}} \vert 2\pi \bk \vert^{2k} \overline{\bu_\bk^*} \cdot \bg_\bk,
        \end{aligned}
    \]
    and therefore by the Cauchy-Schartz inequality,
    \begin{align*}
        \mu \sum_{\bk \in \cR_0^*\setminus\{\bzero\}} & \vert 2\pi\bk \vert^{2k+2} \vert \bu^*_\bk \vert^2 \leq \frac{1}{4\sin^2(\theta/2)} \sum_{\bk \in \cR_0^*\setminus\{\bzero\}} \left ( \vert 2\pi \bk \vert^{k} \vert \bu^*_\bk\vert \right )  \left ( \vert 2\pi \bk \vert^{k} \vert \bg_\bk \vert \right ) \\
        &\leq \frac{1}{4\sin^2(\theta/2)} \left ( \sum_{\bk \in \cR_0^*\setminus\{\bzero\}}  \vert 2\pi \bk \vert^{2k} \vert \bu^*_\bk\vert^2 \right ) ^{1/2}\left ( \sum_{\bk \in \cR_0^*\setminus\{\bzero\}} \vert 2\pi \bk \vert^{2k} \vert \bg_\bk \vert^2 \right ) ^{1/2}.
    \end{align*}
    Thus $\mu \Vert \nabla^{k+1} \bu^* \Vert_{L^2}^2 \leq \frac{1}{4\sin^2(\theta/2)} \Vert \nabla^{k} \bu^* \Vert_{L^2} \Vert \nabla^k \bg \Vert_{L^2}$.
    Now, by an appropriate Faà di Bruno's formula we can show that there exists a generic constant $c_k > 0$ independent of $\theta$ and $\Phi_0$ such that: $\Vert \nabla^k \bg \Vert_{L^2} \leq c_k \Vert \Phi_0 \Vert_{C^{k+1}} (1 + \Vert \bu^* \Vert_{k,2}$), and thus
    \[
        \Vert \nabla^{k+1} \bu^* \Vert^2_{L^2} \leq \frac{c_k \Vert \Phi_0 \Vert_{C^{k+1} }}{4\mu\sin^2(\theta/2)} \Vert \bu^* \Vert_{k,2} \left (1 + \Vert \bu^* \Vert_{k,2} \right ).
    \]
    Let us now assume that for $k \geq 2$,
    \[
        \Vert \bu^* \Vert_{k,2} \leq \frac{  C_k  }{\sin^k(\theta/2)} ,
    \]
    where the constant $C_k$ depends only on $\mE$, $\mu$, and $\Vert \Phi_0 \Vert_{C^{k}}$. Note that this is already true for $k = 2$. Then we have shown that
    \begin{align*}
        \Vert \bu^* \Vert^2_{k+1,2}
        & \leq \frac{  C^2_k  }{\sin^{2k}(\theta/2)} + \frac{c_k C_k \Vert \Phi_0 \Vert_{C^{k+1} }}{4 \mu \sin^{2+k}(\theta/2)}\left ( 1 + \frac{  C_k  }{\sin^k(\theta/2)} \right )\\
        & \leq \left ( \frac{C_{k+1}}{\sin^{k+1}(\theta/2)} \right )^2,
    \end{align*}
    where $C_{k+1} = \left ( C^2_k  + \frac{c_k (1 + C_k)\Vert \Phi_0 \Vert_{C^{k+1} } }{4 \mu} \right )^{1/2}$. Proceeding inductively, (c) follows.
\end{proof}


\section{Numerical study: example of graphene bilayers}
\label{sec:numerics}

As an illustration of the previous framework and analysis, we propose to study numerically the relaxation of a slightly twisted graphene bilayer system~\cite{Dai2016}. This particular assembly is also known for its exceptional electronic properties~\cite{carr2017,Cao2018superconductivity,Cao2018correlated} which might be influenced by the domain formation due to mechanical relaxation~\cite{KimRelax18}.

\paragraph{Parameters} Graphene corresponds to a triangular lattice structure with the basis
\[
    \mE = \sqrt{3}a_0 \begin{bmatrix}
       \nicefrac{\sqrt{3}}{2} & \nicefrac{\sqrt{3}}{2} \\ -\nicefrac 12 & \nicefrac 12
    \end{bmatrix} \qquad \qquad \text{where} \quad a_0 = 1.42\; \textrm{nm},
\]
such that the twisted lattices are given by~\eqref{eq:TwistLattices}.
The inter-layer GSFE is assumed to take the form~\cite{carrrelax}:
\begin{equation}
    \Phi^\mathrm{misfit}_{1+}(\bgamma_2) = \phi(2\pi\mathrm{E}_2^{-1} \bgamma_2), \qquad \Phi^\mathrm{misfit}_{2-}(\bgamma_1) = \phi(2\pi\mathrm{E}_1^{-1} \bgamma_1),
\end{equation}
where the symmetry-adapted functional $\phi$ is defined periodically on $[0, 2\pi)^2$ as:
\[
    \begin{aligned}
    \phi\left ( \begin{bmatrix}
            v \\ w
        \end{bmatrix} \right ) := c_0 & + c_1 \left [ \cos(v) + \cos(w) + \cos(v+w)\right ] \\
        & + c_2 \left [ \cos(v+2w) + \cos(v-w) + \cos(2v+w)\right ] \\
        & + c_3 \left [ \cos(2v) + \cos(2w) + \cos(2v+2w)\right ].
    \end{aligned}
\]
The Lamé parameters $\lambda,$ $\mu$~\eqref{def:LinearElasticFunctional} of the graphene sheets as well as the GSFE coefficients $c_{0-3}$ accurately fitted from vdW-DFT calculations~\cite{carrrelax} are summed up in Table~\ref{table:GrapheneCoefficients}.
\begin{table}[t]
    \centering
    \begin{tabular}{|r | r || r | r | r | r |}
        \hline
        $\lambda$ & $\mu$ & $c_0$ & $c_1$ & $c_2$ & $c_3$ \\ \hline
        37,950 & 47,352 & 6.832 & 4.064 & -0.374 & -0.095\\
        \hline
      \end{tabular}
      \caption{Elastic moduli and GSFE coefficients for graphene bilayers in units of meV/unit cell area.}
      \label{table:GrapheneCoefficients}
\end{table}

\paragraph{Discretization and numerical scheme} For consistency, we discretize here directly the minimization problem~\eqref{def:CBModel}. Let us define a uniform $N \times N$ grid on the torus $[-1/2,1/2)^2$:
\begin{equation}
    \mathcal{G}_N := \left \{ \frac{\delta - \left \lfloor \nicefrac{N}{2} \right \rfloor} N, \dots, \frac {-1} N, 0, \frac 1 N, \dots, \frac{\left \lfloor \nicefrac{N}{2} \right \rfloor}{N} \right \}^2, \quad \text{where } \delta = \begin{cases}
        0, & N \text{ odd}, \\ 1, & N \text{ even}.
    \end{cases}
\end{equation}
We introduce next the set of unknown nodal values $\bU^N = \{ \bU^N_\bn \}_{\bn \in \mathcal{G}_N} $ such that nodal values of the displacements $\bu_1$ and $\bu_2$ on $\Gamma_2$ and $\Gamma_1$ respectively are given by
\[
    \bu_1^N(\mE_2 \bn) = \bU^N_\bn = -\bu_2^N(-\mE_1 \bn) \qquad \text{for } \bn \in \mathcal{G}_N,
\]
where we have used the symmetry~\eqref{eq:EulerLagrangeStrongRef} to reduce the number of free variables. We next interpolate these values by the Fourier series:
\begin{equation}\label{eq:FourierDiscretization}
    \bu_1^N(\bgamma_2) = \sum_{\bk \in \mathcal{G}_N} \widehat{\bU}^N_{\bk} e^{\frac{2i\pi}{N} \bk \cdot \mE_2^{-1}\bgamma_2}, \qquad
    \bu_2^N(\bgamma_1) = - \sum_{\bk \in \mathcal{G}_N} \widehat{\bU}^N_{\bk} e^{-\frac{2i\pi}{N}\bk \cdot \mE_1^{-1}\bgamma_1}.
\end{equation}
Note that the Fourier coefficients appearing in this expression can be computed efficiently from the nodal values $\bU^N$ by the Fast Fourier Transform.
The elastic energy~\eqref{def:LinearElasticFunctional} is then exactly computable, while the misfit energy~\eqref{def:MisfitEnergy} can be approximated by uniform quadrature, a straightforward calculation leading to:
\begin{equation}\label{eq:DiscreteMisfitEnergy}
        E_N^\mathrm{misfit}(\bU) := \frac {1}{2 N^2}  \sum_{\bn \in \mathcal{G}_N}\left [  \phi \left( 2\pi \left ( \bn - 2 \mE_2^{-1} \bU^N_\bn \right ) \right)  + \phi \left(2\pi \left ( \bn - 2 \mE_1^{-1} \bU^N_\bn \right ) \right) \right ].
\end{equation}
An explicit expression can also be obtained for the gradient of these energies.
We implemented this approach in the Julia language~\cite{bezanson2017julia}, using the limited-memory BFGS quasi-Newton algorithm from the \texttt{Optim.jl} Julia library to minimize numerically the resulting total energy.

\paragraph{Results and Discussion}
We first discuss numerical results obtained for the twist angle $\theta = 0.3^\circ$, presented in Figure~\ref{fig:hull_space}. The grid size is chosen as $N = 144$.
\begin{figure}[t]
    \centering
    \begin{subfigure}[t]{.45\textwidth}
        \centering
        \includegraphics[height=1.3in]{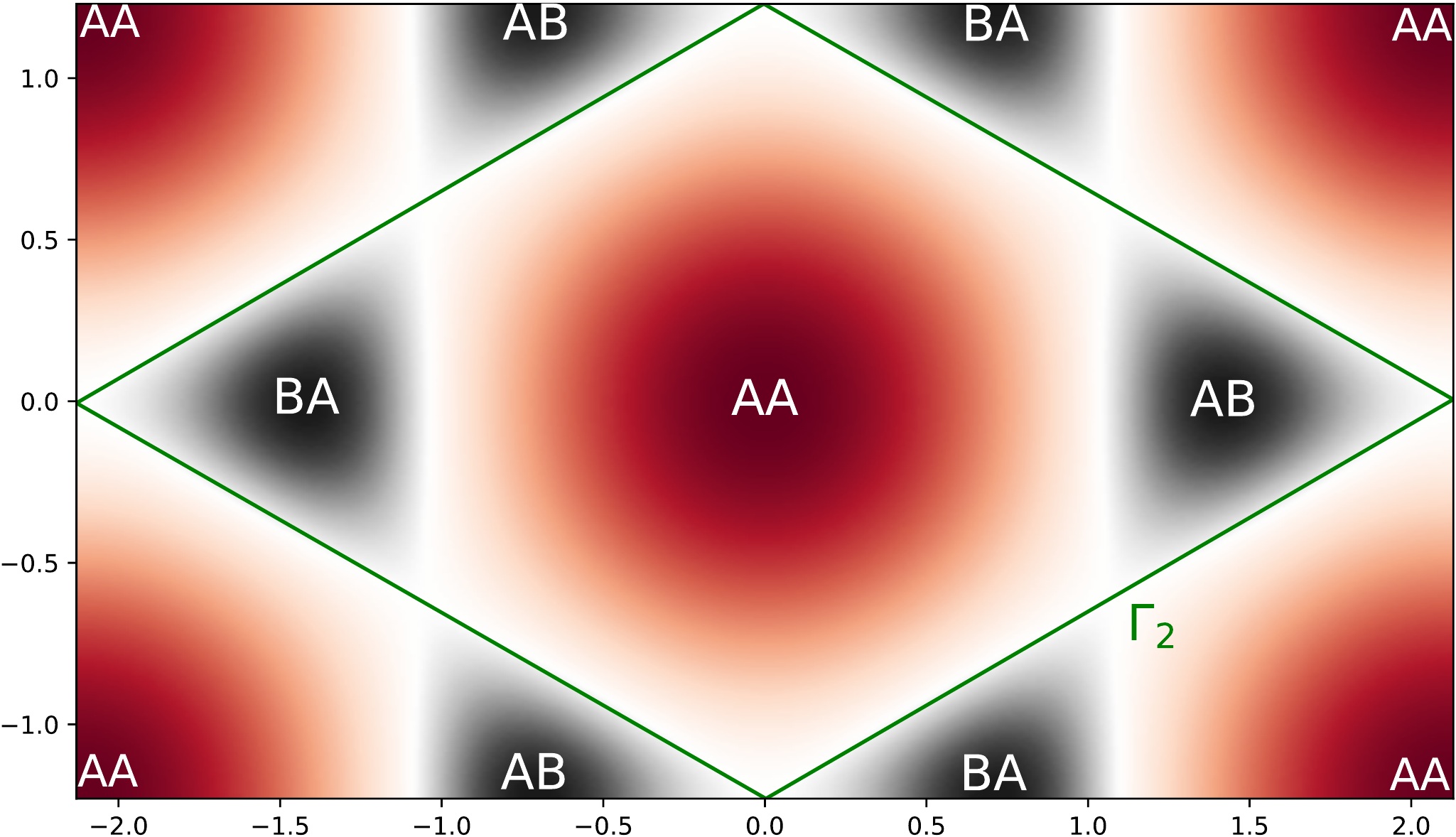}
        \caption{GSFE landscape before relaxation}\label{fig:hull_space_a}
    \end{subfigure}
    \begin{subfigure}[t]{.54\textwidth}
        \centering
        \includegraphics[height=1.326in]{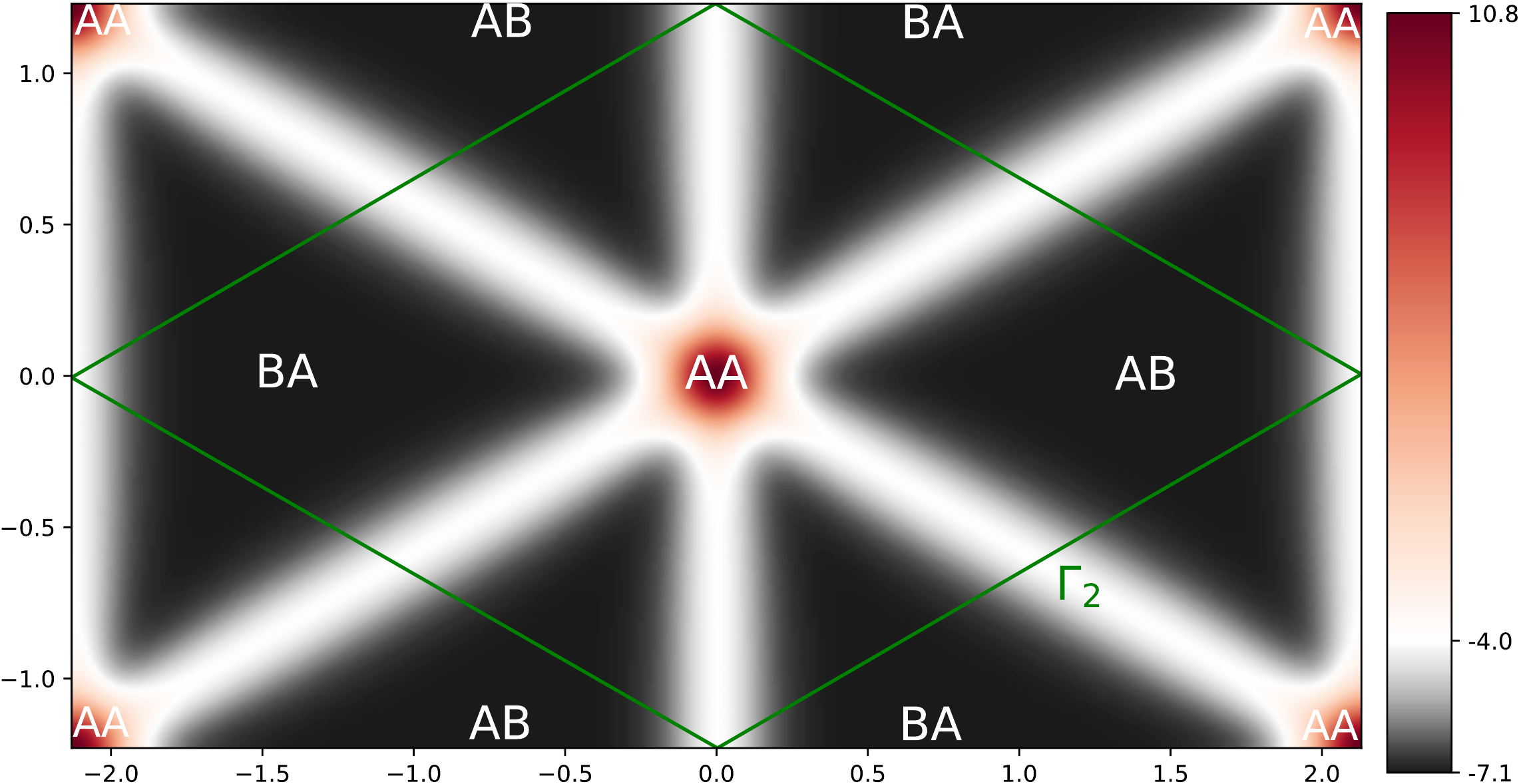}
        \caption{GSFE landscape after relaxation}\label{fig:hull_space_b}
    \end{subfigure}

    \begin{subfigure}[t]{.455\textwidth}
        \centering
        \includegraphics[height=1.3in]{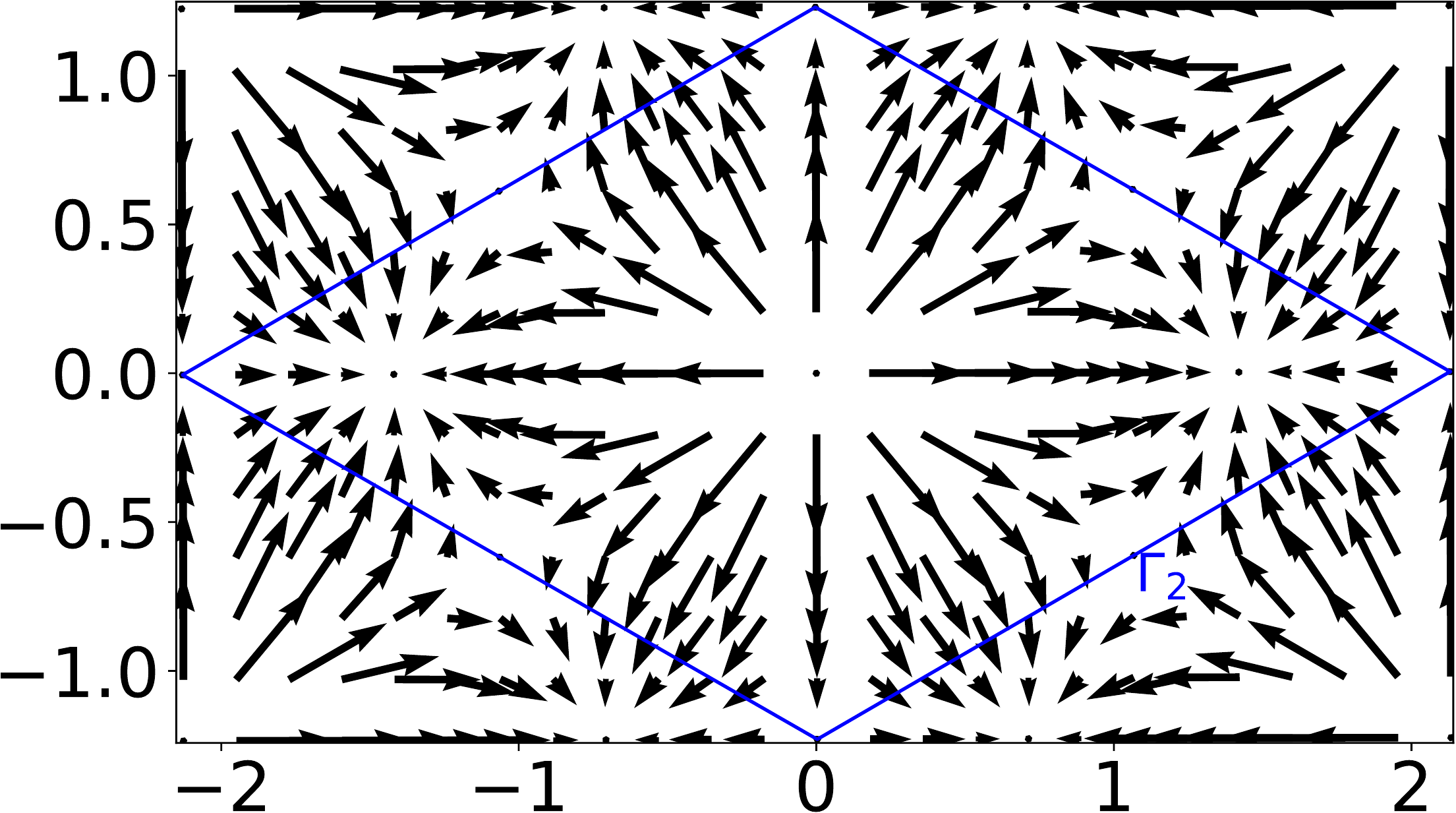}
        \caption{Displacement $-\bu_1(\bgamma_2)$ on the unit cell}
        \label{fig:hull_space_c}
    \end{subfigure}
    \begin{subfigure}[t]{.535\textwidth}
        \centering
        \includegraphics[height=1.3in]{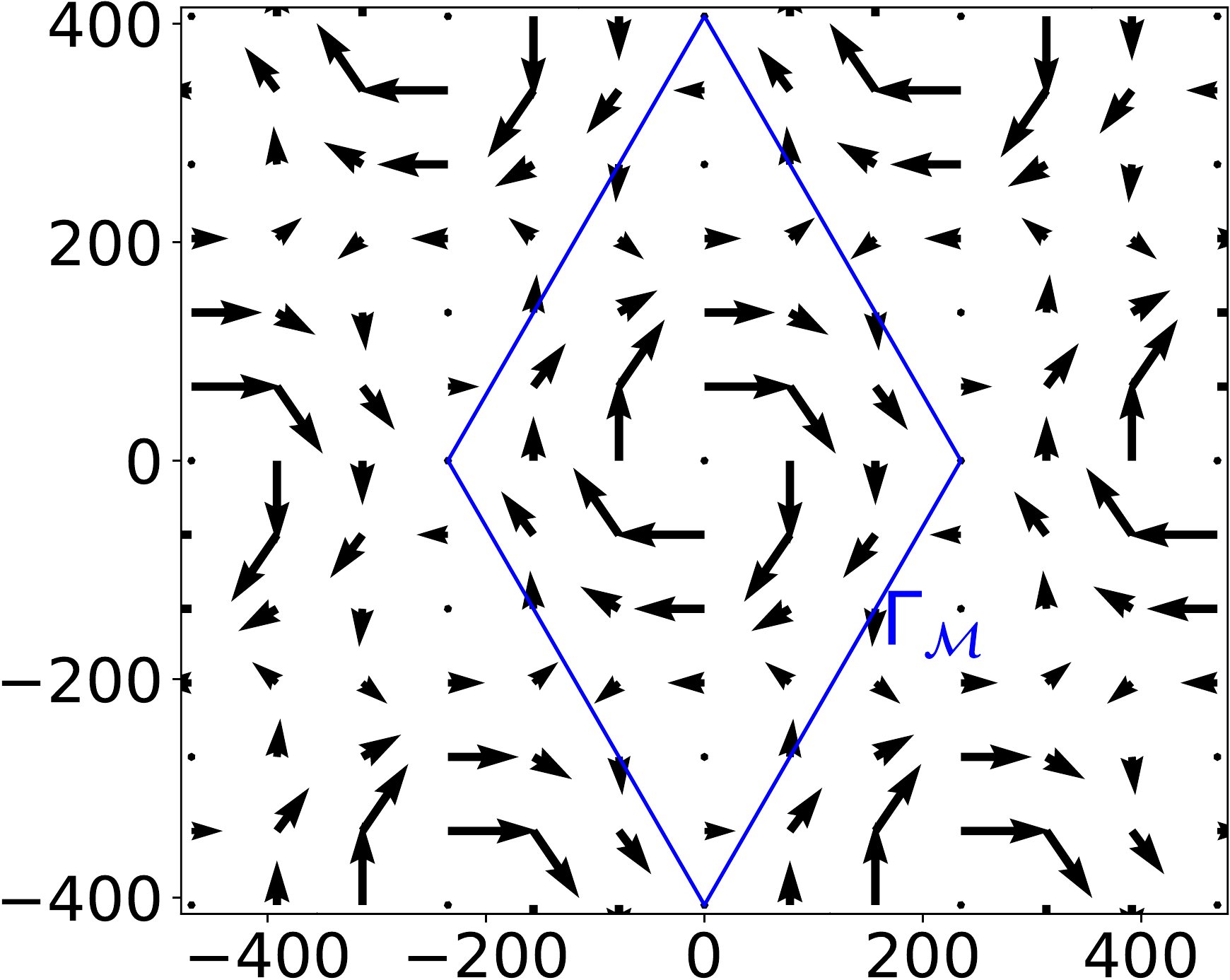}
        \caption{Displacement $\bu_{\mathcal{M},1}(\bx)$ on the moiré cell}
        \label{fig:hull_space_d}
    \end{subfigure}
    \caption{Relaxation for a graphene bilayer with a $0.3^\circ$ relative twist angle. (a): GSFE functional $\Phi_{1+}$ in the reference configuration space $\Gamma_2$, showing the maxima and minima at shifts corresponding respectively to AA and AB/BA stackings. (b) Relaxed local misfit energy as a function of original shift, $\Phi_{1+}\left (\bgamma_2 - 2 \bu_1(\bgamma_2) \right )$. (c) Displacement of atoms in layer 1 in configuration space $-\bu_1$, corresponding to the change in the local stacking configuration if layer $2$ were not moving. (d) Displacement field of atoms in layer 1 in real space, with the periodicity of the moiré cell. }
    \label{fig:hull_space}
\end{figure}
 We observe the formation of a well-defined triangular pattern in configuration space, see Figure~\ref{fig:hull_space_b}. In real space, this corresponds to the triangular meso-scale pattern observed in experiments results~\cite{KimRelax18}: relaxation causes the expansion of the regions of lowest-energy stacking, where the lattices are staggered (graphite or Bernal stacking), into well-defined triangular domains arranged in a hexagonal pattern, and the contraction of higher-energy stacking regions, in particular where lattices are vertically aligned. Note that expansion or contraction in configuration space (as seen on Figure~\ref{fig:hull_space_c}) corresponds to a rotation in real space (as seen on Figure~\ref{fig:hull_space_d}), such that the effective local twist angle is enhanced in regions of higher energy and reduced in regions of lower energy. Thus, the layers are brought into almost perfect alignment with the AB/BA triangular domains.

Next, let us discuss briefly the atomistic mappings corresponding to the deformations computed above. To account for the hexagonal (multilattice) nature of the graphene lattice, two sublattices denoted by A, B can be introduced with respective shifts $\mp a_0/2 \be_1$ in each of the layers before the twist is taken into account. We then make the simplification that the shift field~\cite{van2013symmetries} is fixed: for each of the two layers, the deformed sublattices from a starting configuration $\omega = (\bgamma_1, \bgamma_2)$ are given by the atomistic deformations:
\begin{align*}\label{def:MoireMappings2}
    \left \{ \begin{aligned}
        \bY^\omega_{1,A}(\br_1) &= \br_{1} - a_0/2 \mR_{-\theta/2} \be_1 + \bu_1(\bgamma_2 - \br_1) , \\
         \bY^\omega_{1,B}(\br_1) &= \br_{1} + a_0/2 \mR_{-\theta/2} \be_1 + \bu_1(\bgamma_2 - \br_1) ,
    \end{aligned}  \right. \qquad \text{for } \br_1 \in \bgamma_1 + \cR_1, \\
    \left \{ \begin{aligned}
        \bY^\omega_{2,A}(\br_1) &= \br_{2} - a_0/2 \mR_{+\theta/2} \be_1 + \bu_2(\bgamma_1 - \br_2), \\
         \bY^\omega_{2,B}(\br_1) &= \br_{2} + a_0/2 \mR_{+\theta/2} \be_1 + \bu_2(\bgamma_1 - \br_2),
    \end{aligned}  \right. \qquad \text{for } \br_2 \in \bgamma_2 + \cR_2.
\end{align*}
These mappings allow us to visualize the domain structures at the atomistic level in real space, as presented in Figure~\ref{fig:moire_relax} for a $3^\circ$ twist angle\footnote{Note that interlayer coupling used in the computation leading to Figure~\ref{fig:moire_relax} was artifically enhanced by a factor of $100$ with respect to the coefficients in~\ref{table:GrapheneCoefficients} such that the scales of both unit and moiré cells are visible while still ensuring significant relaxation.}.

A closer inspection of these structures reveals that the pattern is \textit{not periodic} with the moiré lattice $\Gamma_\mathcal{M}$, although the calculation of the hull functions $\bu_1$ and $\bu_2$ uses periodic boundary conditions, and the moiré cell $\Gamma_\mathcal{M}$ introduced in paragraph~\ref{sec:moirecell} resembles a real-space supercell. In this sense, our method allows us to truly model the incommensurability of the lattices $\cR_1$ and $\cR_2$ without the need for constructing an appropriate supercell as in previous works~\cite{Wijk2015,Dai2016,zhang2017energy}.

\section*{Conclusion}
    We have presented a novel elastostatics model~\eqref{def:CBEnergy}-\eqref{def:CBModel} for the relaxation of an incommensurate multi-layered structure.
    Our derivation is based on the ergodic properties~\eqref{def:discreteBirkhoff} with respect to the translation group of local configurations for incommensurate multi-layers,
    a gradient for the atomistic deformation~\eqref{def:GradAtomDef}, and the modulated local disregistry~\eqref{def:ModDisregistry3}.
    We have presented a rigorous analysis in the bilayer case ($p = 2$), showing how the continuum relaxation problem is well-posed and amenable to numerical simulations.
    We can construct the actual relaxed aperiodic atomistic positions as in~\eqref{eq:atomdeformation}, i.e., by sampling the displacement field at the aperiodic atomistic positions of the unrelaxed incommensurate heterostructure, as seen in Figure~\ref{fig:moire_relax}.
    This provides the link between the atomistic picture and existing continuum bilayer models~\cite{Dai2016}.
    These relaxed atomistic configurations can then be used to compute diffraction patterns~\cite{KimRelax18}, the electronic density of states, and transport properties~\cite{Kubo2017,MassattDOS16}.

    Future perspectives include the further study of the proposed models in particular for $p > 2$ layers, where one must face difficulties due to the lack of ellipticity of the elastic energy functional~\eqref{def:CBEnergy} from the analytical point of view, and the curse of dimensionality from the numerical point of view as the dimension of the configuration space becomes larger than $4$.

\section*{Acknowledgments}
This work was supported in part by ARO MURI Award W911NF-14-1-0247 and by the National Science Foundation under NSF Award DMS-1819220.

\appendix

\section{Proof of Prop.~\ref{prop:Interpolants}} \label{sec:appendix}
In this appendix, we detail the technical proof of Proposition~\ref{prop:Interpolants}, which we recall first for ease of reading.

\begin{proposition}
Let $\bu \in W^{2,q}(X)$, then:
    \begin{equation}\label{est:InterpolantDifference_app}
         \left \Vert \widehat{\bu}_j - \widetilde{\bu}_j \right \Vert_{L^q(\Omega)} \leq \left ( \frac{\vert \Gamma_j \vert}{2q+1} \right )^{1/q} \theta ^2 \; \left \Vert \nabla^2_\omega \bu_j \right \Vert_{L^q(X_j)}
    \end{equation}
    where $X_j = \bigtimes_{i \neq j} \Gamma_i$ is the subset of the transversal corresponding to lattice sites of layer $j$, $\nabla^2_\omega \bu_j$ is understood as a $2$-linear form for which the norm is defined as $\Vert \ell \Vert := \sup_{\vert \mathbf{h}_1 \vert = \vert \mathbf{h}_2 \vert = 1 } \left \vert \ell[\mathbf{h}_1, \mathbf{h}_2] \right \vert $ and
    \begin{equation}\label{def:theta_app}
        \theta = \sqrt{p} \sup_{1 \leq i,j \leq p} \left \Vert \mE_i - \mE_j \right \Vert, \qquad \text{where } \Vert \cdot \Vert \text{ denotes the $\bbR^2$-operator norm}.
    \end{equation}
\end{proposition}

\begin{proof}
    It is natural to relate this result to local finite element interpolant error estimation and follow similar steps. Let $\omega \in \Omega$ be an arbitrary configuration, $j$ be an arbitrary layer number. Consider a Taylor expansion up to degree 1 of $\bu_j$ around the point $\Pi_j\omega$:
    \[
        T^1_\omega \bu_j(\delta \omega) = \bu_j \left ( \Pi_j \omega \right ) + \nabla_\omega \bu_j \left(\Pi_j \omega \right )\cdot \delta \omega \text{ where } \delta \omega = (\boldsymbol{\delta \omega}_1, \ldots, \boldsymbol{\delta \omega}_p) \in \bbR^{2p},\ \boldsymbol{\delta \omega}_j = \bzero.
    \]
    Using the integral formula for the residual of the Taylor series, we have:
    \begin{align*}
        \bu_j(\Pi_j \omega + \delta \omega) - T^1_\omega \bu_j(\delta \omega) &= \int_0^1 \frac{(1-h)^2}{2}\frac{d^2\bu_j}{dh^2} (\Pi_j \omega + h \delta \omega)d h\\
        &= \int_0^1 \frac{(1-h)^2}{2} \nabla_\omega^2 \bu_j (\Pi_j \omega + h \delta \omega) [\delta \omega, \delta \omega] d h,
    \end{align*}
    so by Jensen's inequality we obtain the estimate:
    \begin{equation}\label{est:Taylor}
        \Vert \bu_j(\Pi_j \omega + \delta \omega) - T^1_\omega \bu_j(\delta \omega) \Vert^q \leq  \Vert \delta \omega \Vert^{2q} \int_0^1  \frac{(1-h)^{2q}}{2^q} \left \Vert \nabla_\omega^2 \bu_j (\Pi_j \omega + h \delta \omega) \right \Vert^q d h.
    \end{equation}
    As earlier (see~\eqref{def:InterpolantDisplacement}), let us write $\bgamma_j = \mE_j \begin{bmatrix} s \\ t \end{bmatrix}$ with $0 \leq s, t < 1$.
The lattice sites in layer $j$ around the origin are located at the four points $\br_{ab}$ with $a, b \in \{ 0,1 \}$ defined as in~\eqref{def:Corners}.
One checks easily from the definition~\eqref{def:InterpolantConfiguration} that
\[
    \mathtt{T}_{-\br_{ab}} \omega = \Pi_j \omega + \delta \omega_{ab} \qquad \text{for } a,b \in \{0,1\},
\]
where the shifts $\delta \omega_{00}$, $\delta \omega_{10}$, $\delta \omega_{01}$, $\delta \omega_{11}$ can be chosen as
\begin{align*}
    & \delta \omega_{00} = \left ( (\mE_1 - \mE_j) \begin{bmatrix} s\\t \end{bmatrix}, \ldots, (\mE_p - \mE_j) \begin{bmatrix} s\\t \end{bmatrix} \right ), \\
    & \delta \omega_{10} = \left ( (\mE_1 - \mE_j) \begin{bmatrix} s-1\\t \end{bmatrix}, \ldots, (\mE_p - \mE_j) \begin{bmatrix} s-1\\t \end{bmatrix} \right ), \\
    & \delta \omega_{01} = \left ( (\mE_1 - \mE_j) \begin{bmatrix} s\\t-1 \end{bmatrix}, \ldots, (\mE_p - \mE_j) \begin{bmatrix} s\\t-1 \end{bmatrix} \right ), \\
    & \delta \omega_{11} = \left ( (\mE_1 - \mE_j) \begin{bmatrix} s-1\\t-1 \end{bmatrix}, \ldots, (\mE_p - \mE_j) \begin{bmatrix} s-1\\t-1 \end{bmatrix} \right )
\end{align*}
since $\delta \omega_{ab}$ is defined on $\Omega$ and is thus invariant under lattice shifts.
Note that for $a,b \in \{0,1\}$,
\[
    \delta \omega_{ab,j} = \bzero \qquad\text{and}\qquad \Vert \delta \omega_{ab} \Vert \leq \sqrt{2}\;  \theta,
\]
where $\theta$ is defined in~\eqref{def:theta}.
Furthermore, the weighted average of these shifts is zero:
\[
    \sum_{a,b \in \{0, 1 \}}\alpha_{ab}\; \delta \omega_{ab}  = 0,
\]
where we have introduced the bilinear weights
\[
    \alpha_{00} = (1-s)(1-t), \quad \alpha_{10} = s(1-t), \quad \alpha_{01} = (1-s)t, \quad \alpha_{11} = st.
\]
As a consequence, by the affine character of the Taylor approximant $T_\omega^1 \bu_j$ defined above,
\begin{equation}\label{id:SmoothInterpolantTaylor}
    \widehat{\bu}_j(\omega) = T^1_\omega \bu_j(0) = \sum_{a,b \in \{ 0, 1 \} } \alpha_{ab} \; T^1_\omega \bu_j(\delta \omega_{ab}).
\end{equation}
Let us now rewrite the definition~\eqref{def:InterpolantDisplacement} of the bilinear interpolant $\widetilde{\bu}_j$ as:
\begin{equation}\label{id:DiscreteInterpolantShifts}
    \widetilde{\bu}_j(\omega) = \sum_{a,b \in \{ 0, 1 \} } \alpha_{ab} \; \bu_j(\Pi_j \omega + \delta \omega_{ab}).
\end{equation}
Taking the difference of the identities~\eqref{id:SmoothInterpolantTaylor},~\eqref{id:DiscreteInterpolantShifts} and using convexity of the norm and the above Taylor estimate~\eqref{est:Taylor}, we find the pointwise estimate
\begin{equation*}
    \begin{aligned}
      \left \Vert \widehat{\bu}_j(\omega) - \widetilde{\bu}_j(\omega) \right \Vert^q  \leq \theta^{2q} \int_0^1 (1-h)^{2q}
                \sum_{a,b \in \{ 0, 1 \} } \alpha_{ab} \;  \left \Vert \nabla_\omega^2 \bu_j (\Pi_j \omega + h \delta \omega_{ab}) \right \Vert^q d h.
    \end{aligned}
\end{equation*}
We may now integrate over the configuration parameter $\omega = (\bgamma_1, \ldots, \bgamma_p)$. The difference between $\omega$ and $\left ( \Pi_j \omega + h \delta \omega_{ab}\right )$ depends only on $s$, $t$ and $h$, for example
\[
    \omega - \left ( \Pi_j \omega + h \delta \omega_{00}\right ) = - \left ( \left ( h \mE_j + (1-h) \mE_1  \right) \begin{bmatrix} s \\ t \end{bmatrix} , \ldots, \left ( h \mE_j + (1-h) \mE_p  \right) \begin{bmatrix} s \\ t \end{bmatrix}  \right ).
\]
Integrating over the variables $\bgamma_i$ for $i \neq j$ with fixed values of $\bgamma_j= \mE_j \begin{bmatrix} s \\ t \end{bmatrix}$ and $h$ we find that by translation invariance of the Lebesgue measure,
\[
    \idotsint_{\bigtimes_{i \neq j} \Gamma_i } \left \Vert \nabla_\omega^2 \bu_j (\Pi_j \omega + h \delta \omega_{ab}) \right \Vert^q_\mathrm{op} \prod_{i \neq j} d \bgamma_i = \left \Vert \nabla^2_\omega \bu_j \right \Vert_{L^q(X_j, \Vert \cdot \Vert_\mathrm{op})}^q.
\]
This leads to
\begin{equation*}
    \begin{aligned}
      \idotsint_{\bigtimes_{i \neq j} \Gamma_i } \left \Vert \widehat{\bu}_j(\omega) - \widetilde{\bu}_j(\omega) \right \Vert^q \prod_{i \neq j} d \bgamma_i \leq \theta^{2q}
                \left \Vert \nabla^2_\omega \bu_j \right \Vert_{L^q(X_j, \Vert \cdot \Vert_\mathrm{op})}\int_0^1 (1-h)^{2q}  d h.
    \end{aligned}
\end{equation*}
Since the right-hand side does not depend on the remaining variable $\bgamma_j$, one last integration over it yields:
\begin{equation*}
      \left \Vert \widehat{\bu}_j(\omega) - \widetilde{\bu}_j(\omega) \right \Vert_{L^q(\Omega)}^q \leq \frac{\vert \Gamma_j \vert }{2q+1} \theta^{2q} \; \left \Vert \nabla^2_\omega \bu_j \right \Vert_{L^q(X_j, \Vert \cdot \Vert_\mathrm{op})}^q,
\end{equation*}
which proves the desired estimate~\eqref{est:InterpolantDifference}.
\end{proof}

\end{document}